\documentclass{fundam}

\makeatletter
\let\c@definition\relax
\makeatother

\usepackage[T1]{fontenc}
\usepackage[USenglish]{babel}

\usepackage{fontawesome5}

\usepackage{amsmath,amssymb}
\usepackage{mathtools}
\usepackage{shuffle}

\usepackage{tabularray}
\UseTblrLibrary{booktabs}
\usepackage{subcaption}

\usepackage{tikz}
\usetikzlibrary{
	shapes.multipart,
	arrows.meta,
	petri,
	positioning,
	calc,
	decorations.pathreplacing,
}

\usepackage{environ}

\usepackage[autostyle=true]{csquotes}
\makeatletter
\let\c@author\relax
\makeatother
\usepackage[
	backend=biber,
	style=numeric-comp,
	sorting=nyvt,
	sortcites=true,
	uniquename=init,
	giveninits=true,
	useprefix,
	maxbibnames=99,
	maxcitenames=2,
	minnames=1,
]{biblatex}
\DeclareSortingNamekeyTemplate{
	\keypart{\namepart{family}}
	\keypart{\namepart{prefix}}
	\keypart{\namepart{given}}
	\keypart{\namepart{suffix}}
}
\DeclareFieldFormat{eprint:urn}{%
	URN\addcolon\space
	\ifhyperref
	{\href{http://www.nbn-resolving.org/#1}{\nolinkurl{#1}}}
	{\nolinkurl{#1}}
}
\DeclareSourcemap{
	\maps[datatype=bibtex]{
		\map[overwrite]{
			\step[fieldsource=shortseries, fieldtarget=series]
			\step[fieldset=eventtitle, null]
			\step[fieldset=eventdate, null]
			\step[fieldset=venue, null]
		}
		\map[overwrite]{
			\step[fieldsource=doi, final]
			\step[fieldset=issn, null]
			\step[fieldset=url, null]
			\step[fieldset=eprint, null]
		}
	}
}
\usepackage{hyperref}
\usepackage{zref-clever}
\zcsetup{noabbrev, cap, nameinlink=false}
\zcLanguageSetup{english}{
	type = theorem ,
	Name-sg-ab = Thm. ,
	name-sg-ab = thm. ,
	Name-pl-ab = Thms. ,
	name-pl-ab = thms. ,
	type = lemma ,
	Name-sg-ab = Lem. ,
	name-sg-ab = lem. ,
	Name-pl-ab = Lems. ,
	name-pl-ab = lems. ,
	type = corollary ,
	Name-sg-ab = Cor. ,
	name-sg-ab = cor. ,
	Name-pl-ab = Cors. ,
	name-pl-ab = cors. ,
	type = proposition ,
	Name-sg-ab = Prop. ,
	name-sg-ab = prop. ,
	Name-pl-ab = Props. ,
	name-pl-ab = props. ,
}
\usepackage{orcidlink}

\newcommand*{\lbfc}{\textsc{lbfc}}

\newcommand*{\M}{\mathcal{M}}
\newcommand*{\N}{\mathbb N}

\newcommand*{\Q}{\mathbb Q}
\newcommand*{\costRange}{\Q_{\geq0}}
\renewcommand*{\vec}{\boldsymbol}
\newcommand*{\mat}[1]{\mathbf{#1}}
\DeclareMathOperator*{\seqsum}{\textstyle\sum}
\newcommand*{\pset}{{\bullet}}
\newcommand*{\nomove}{{\gg}}
\newcommand*{\init}{\mathit{init}}
\newcommand*{\final}{\mathit{final}}
\newcommand*{\aux}{\mathit{aux}}
\newcommand*{\Sigmatau}{\Sigma\cup\{\tau\}}
\makeatletter
\newcommand*{\L@Icon}[1]{\text{\raisebox{\depth}{\resizebox*{!}{1.5ex}{\faIcon{#1}}}}}
\newcommand*{\Lstart}{\L@Icon{play}}
\newcommand*{\Lstop}{\L@Icon{stop}}
\newcommand*{\Lskip}{\L@Icon{fast-forward}}
\newcommand*{\Lreset}{\L@Icon{undo}}
\newcommand*{\Lacc}{\L@Icon{check}}
\newcommand*{\Lrej}{\L@Icon{times}}
\newcommand*{\Laux}{\L@Icon{h-square}}
\makeatother
\DeclarePairedDelimiter\abs\lvert\rvert
\DeclarePairedDelimiter\seq\langle\rangle
\providecommand\given{}
\makeatletter
\newcommand\@SetSymbol[1][]{%
	\nonscript\:#1\vert
	\allowbreak
	\nonscript\:
	\mathopen{}}
\DeclarePairedDelimiterX\Set[1]\lbrace\rbrace{%
	\renewcommand\given{\@SetSymbol[\delimsize]}
	#1
}
\DeclarePairedDelimiterX\Multiset[1]\lbrack\rbrack{%
	\renewcommand\given{\@SetSymbol[\delimsize]}
	#1
}
\makeatother
\DeclarePairedDelimiter{\supp}{\langle}{\rangle}
\DeclarePairedDelimiter{\fire}{\lbrack}{\rangle}
\DeclarePairedDelimiter{\reach}{\lbrack}{\rangle}
\DeclarePairedDelimiterXPP{\scfs}[1]{\phi}{(}{)}{}{#1}
\DeclarePairedDelimiterXPP{\lang}[1]{\mathcal L}{(}{)}{}{#1}
\DeclarePairedDelimiterXPP{\reachGraph}[1]{\mathcal R}{(}{)}{}{#1}
\makeatletter
\MHInternalSyntaxOn
\def\MH_bigshuffle_scaler:N #1{%
	\vcenter{\hbox{#1$\m@th\mkern-.5mu\shuffle\mkern-.5mu$}}}
\def\MH_bigshuffle_inner: {
	\mathchoice{\MH_bigshuffle_scaler:N \huge}         % display style
	{\MH_bigshuffle_scaler:N \LARGE}        % text style
	{\MH_bigshuffle_scaler:N {}}            % script style
	{\MH_bigshuffle_scaler:N \footnotesize} % script script style
}
\def\MH_csym_bigshuffle: {\mathop{\MH_bigshuffle_inner:}\nolimits}
\AtBeginDocument{
	\providecommand\bigshuffle{\MH_csym_bigshuffle:}
}
\MHInternalSyntaxOff
\makeatother

\DeclareMathOperator{\ptseq}{\rightarrow}
\DeclareMathOperator{\ptxor}{\times}
\DeclareMathOperator{\ptpar}{\wedge}
\DeclareMathOperator{\ptloop}{\circlearrowleft}

\makeatletter
\newcommand*{\@complexity}[1]{\mathsf{#1}}
\newcommand*{\Lspace}{\@complexity{L}}
\newcommand*{\NL}{\@complexity{NL}}

\renewcommand*{\P}{\@complexity{P}}
\newcommand*{\NP}{\@complexity{NP}}
\newcommand*{\PSPACE}{\@complexity{PSPACE}}
\newcommand*{\EXPSPACE}{\@complexity{EXPSPACE}}
\newcommand*{\Ackermann}{\@complexity{Ackermann}}
\makeatother

\makeatletter
\newcommand*{\@compprob}[1]{\ensuremath{\textup{\textsc{#1}}}}
\newcommand*{\CPalign}{\@compprob{Align}}
\newcommand*{\CPreach}{\@compprob{Reach}}
\newcommand*{\CPcostReach}{\@compprob{MinCostReach}}
\newcommand*{\CPmember}{\@compprob{Member}}
\newcommand*{\CPwordshuffle}{\@compprob{WordShuffle}}
\newcommand*{\CPmembershuffle}{\@compprob{ShuffleMember}}
\makeatother

\newenvironment{proofof}[1]
{\trivlist\PRstyle\item[]{\bfseries Proof of #1:}\newline}{\QED\endtrivlist}
\def\squareforqed{\hbox{\rlap{$\sqcap$}$\sqcup$}}
\def\QED{\ifmmode\squareforqed\else{\unskip\nobreak\hfil
		\penalty50\hskip1em\null\nobreak\hfil\squareforqed
		\parfillskip=0pt\finalhyphendemerits=0\endgraf}\fi}
\NewEnviron{problem}[1]{%
	\begin{center}\fbox{\parbox{.9\columnwidth}{%
				{\centering\scshape #1\par}%
				\parskip=1ex
				\everypar{\hangindent=1em}%
				\BODY
	}}\end{center}
}

\newcommand{\THstyleIt}{\theorembodyfont{\itshape}\theoremheaderfont{\normalfont\bfseries}}
\theoremstyle{plain}
\THstyle
\newtheorem{definition}{Definition}[section]
\THstyleIt
\newtheorem{theorem}{Theorem}[section]
\THstyle
\newtheorem{fact}{Fact}[section]
\THstyleIt
\newtheorem{lemma}{Lemma}[section]
\THstyle
\newtheorem{example}{Example}[section]
\THstyleIt
\newtheorem{assumption}{Assumption}[section]
\newtheorem{proposition}{Proposition}[section]
\newtheorem{remark}{Remark}[section]
\newtheorem{corollary}{Corollary}[section]
\newtheorem{claim}{Claim}[section]

\makeatletter
\newlength\PN@BaseSize
\tikzset{
	node distance=4\PN@BaseSize,
	every label/.append style={
		inner sep=0,
		text height=.7\PN@BaseSize,
		text depth=.3\PN@BaseSize,
	},
	every place/.style={
		draw,
		thick,
		minimum size=1.5\PN@BaseSize,
	},
	every transition/.append style={
		draw,
		thick,
		minimum height=1.5\PN@BaseSize,
		minimum width=1.5\PN@BaseSize,
		text height=.7\PN@BaseSize,
		text depth=.3\PN@BaseSize,
	},
	tau-transition/.style={
		transition,
		fill,
		minimum width=.75\PN@BaseSize,
	},
	labels/.code={\tikzset{
			place/.append style={label=#1:##1},
			transition/.append style={label=#1:##1},
			tau-transition/.append style={label=#1:##1},
		}
	},
	labels/.default=below,
	pre/.style={Latex-,thick},
	post/.style={-Latex,thick},
	final/.style={double},
	position/.style args={#1 degrees from #2}{%
		at=(\iftikz@lib@ignore@size#2\else#2.#1\fi),%
		shift=(#1:\tikz@node@distance),%
	},
	position/.append code={
		\iftikz@lib@ignore@size\pgfkeysalso{/tikz/anchor=center}\else\pgfkeysalso{/tikz/anchor={#1+180}}\fi
	},
	between/.code args={#1 and #2}{
		\iftikz@lib@ignore@size
		\pgfkeysalso{/tikz/at=($(#1)!.5!(#2)$)}
		\else
		\pgfmathanglebetweenpoints{\pgfpointanchor{#1}{center}}{\pgfpointanchor{#2}{center}}
		\pgfcoordinate{__middle}{\pgfpointlineattime{.5}{\pgfpointanchor{#1}{\pgfmathresult}}{\pgfpointanchor{#2}{\pgfmathresult+180}}}
		\pgfkeysalso{/tikz/at=(__middle)}
		\fi
	},
}
\makeatother

\begin{document}
	\title{Computational Complexity of Alignments}
	\author{Christopher T. Schwanen\corresponding \orcidlink{0000-0002-3215-7251}\\
		Chair of Process and Data Science (PADS)\\
		RWTH Aachen University
		\and Wied Pakusa \orcidlink{0009-0004-6302-4445}\\
		Faculty of Mathematics, Informatics and Technology\\
		Koblenz University of Applied Sciences
		\and Wil M. P. van der Aalst \orcidlink{0000-0002-0955-6940}\\
		Chair of Process and Data Science (PADS)\\
		RWTH Aachen University}
	\address{schwanen@pads.rwth-aachen.de}
	\maketitle
	\runninghead{C.\,T.~Schwanen, W.~Pakusa, and W.\,M.\,P.~van~der~Aalst}{Computational Complexity of Alignments}
	
	\begin{abstract}
		In process mining, alignments quantify the degree of deviation between an observed event trace and a business process model and constitute one of the most important conformance checking techniques.
		We study the algorithmic complexity of computing alignments over important classes of Petri nets.
		First, we show that the alignment problem is PSPACE-complete on the class of safe Petri nets and also on the class of safe and sound workflow nets.
		For live, bounded, free-choice systems, we prove the existence of optimal alignments of polynomial length which positions the alignment problem in NP for this class. We further show that computing alignments is NP-complete even on basic subclasses such as process trees and T-systems. We establish NP-completeness on several related classes as well, including acyclic systems.
		Finally, we demonstrate that on S-systems, the complexity of the alignment problem even depends on the number of tokens present. If the number of tokens is bounded by a constant (e.g., on live, safe S-systems), alignments can be computed in polynomial time, whereas the problem is NP-complete if the number of tokens is part of the input, even for safe or live S-systems.
	\end{abstract}
	
	\section{Introduction}
	
	In process-centric applications, log data comes in form of \emph{event logs} which consist of \emph{sequences} of \emph{events}, called \emph{traces}. Each trace, in turn, corresponds to a particular execution of a business process. The events are associated with timestamps and activity identifiers.
	Given this particular structure, classic data analysis approaches turn out to be inadequate for this particular kind of data.
	The objective of \emph{process mining}~\autocite{vanderAalst2016} is to explore and develop specialized data-driven techniques to effectively extract meaningful insights from event logs.
	
	A seminal subfield of process mining, known as \emph{conformance checking}~\autocite{CarmonaDSW2018}, investigates methods for comparing observed behavior of a process, usually given in the form of an event log, against a reference model, typically given as Petri net or BPMN diagram.
	Clearly, in almost all organizations and application domains it is important to ensure compliance with the intended process design.
	Conformance checking aims at finding deviations between the observed behavior and the given model which allows for identifying inefficiencies, regulatory violations, or, on the positive side, for discovering best practices that have not been thought of when the process was designed in the first place.
	Leading vendors of process mining tools have also integrated conformance checking algorithms into their products which underpins the significance of developing efficient and effective conformance checking algorithms in academia.
	While different proposals for conformance metrics have been made during the last decades, as of today, \emph{alignments} \autocite{Adriansyah2014} are considered to be the state-of-the-art technique in conformance checking. We refer to~\autocite{CarmonaDSW2018} for a thorough introduction to this field.
	
	The system in \zcref{fig:align:example} specifies a process with two observable events $a$ and~$b$ in the form of a Petri net.
	We view a Petri net as both, a formal model defining a process, and as a \emph{language acceptor}. In both cases, we agree on an \emph{initial} and a \emph{final} marking and consider runs of the model (so called \emph{firing sequences}) from the initial to the final marking. For each run, the transition labels form a finite \emph{word}, also called a \emph{trace}. Transitions can also be unlabeled (those are called \emph{silent}). The collection of generated words constitutes the \emph{language} accepted by the Petri net.
	For the Petri net in \zcref{fig:align:example}, the initial marking is highlighted (one single token in place $p_\init$) and the intended final marking has a single token in $p_\final$.
	Relative to these markings, the Petri net accepts the traces $\seq{a,a,b,b}$ and $\seq{a,b,a,b}$, but not $\seq{a,b,a,a}$. The accepted language is $(aab | aba)^+ b$.
	
	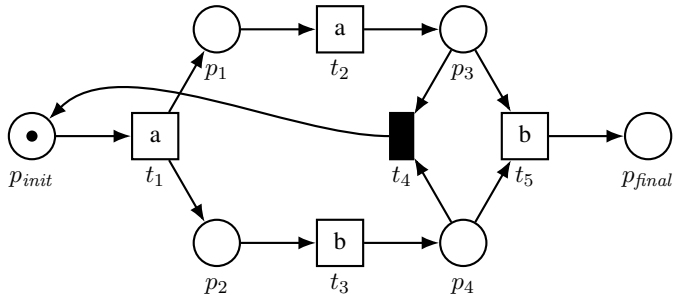
\begin{figure}
		\centering
		\begin{tikzpicture}[scale=.85, transform shape, on grid, labels]
			\node[place={$p_\init$},      tokens=1]                     (pi)  {};
			\node[transition={$t_1$},     position=  0 degrees from pi] (t1) {a} edge[pre] (pi);
			\node[place={$p_1$},          position= 60 degrees from t1] (p1)  {} edge[pre] (t1);
			\node[place={$p_2$},          position=-60 degrees from t1] (p2)  {} edge[pre] (t1);
			\node[transition={$t_2$},     position=  0 degrees from p1] (t2) {a} edge[pre] (p1);
			\node[transition={$t_3$},     position=  0 degrees from p2] (t3) {b} edge[pre] (p2);
			\node[place={$p_3$},          position=  0 degrees from t2] (p3)  {} edge[pre] (t2);
			\node[place={$p_4$},          position=  0 degrees from t3] (p4)  {} edge[pre] (t3);
			\node[transition={$t_5$},     position=-60 degrees from p3] (t5) {b} edge[pre] (p3) edge[pre] (p4);
			\node[tau-transition={$t_4$}, position=240 degrees from p3] (t4)  {} edge[pre] (p3) edge[pre] (p4) edge[post, out=180, in=45, looseness=.75] (pi);
			\node[place={$p_\final$},     position=  0 degrees from t5] (pf)  {} edge[pre] (t5);
		\end{tikzpicture}
		\caption{Example of a Petri net (sound free-choice workflow net)}
		\label{fig:align:example}
	\end{figure}
	
	Dissimilarity is quantified by \enquote{\emph{aligning}} a given trace against the model, which means that we \emph{insert} and \emph{delete} activities into and from the trace until it fits with the model. The goal is to find an \emph{optimal alignment} which minimizes the number of required change operations.
	For example, the trace $\sigma=\seq{a,b,a,a}$ deviates from the model and would incur non-zero costs in an optimal alignment.
	There are two change operations: \emph{insertions} and \emph{deletions} of activities in the given trace.
	This resembles classic \emph{edit distances} studied in formal language theory.
	However, for alignments it is not necessarily the case that all change operations incur the same amount of costs. This is motivated by the business context: for example, it might be unproblematic to remove certain activities from the observed trace while the removal of others would have dramatic impact on the correctness of the process. Of course, the same might be true for certain parts of the normative model.
	The costs of the change operations can also be made dependent on the specific activity label and the current transition of the normative model.
	This is formalized as follows: we assume that the business process and the trace are executed and generated concurrently.
	While the model evolves from its initial to its final marking, and, concurrently, the trace from its first to its last event, there are three different types of \emph{moves}:
	\begin{enumerate}
		\item the system \emph{and} the trace take one step \emph{synchronously}, i.e., the system fires a transition $t$ labeled with $a$ and the trace moves on via its next letter $a$; this is called a \emph{synchronous} move $(a, t)$;
		\item the system fires $t$, but the trace maintains its state; this corresponds to an \emph{insert} operation where the label of $t$ (which might be empty) is inserted into the trace; such move is called \emph{model move} and is denoted by $(\nomove, t)$ where $\nomove$ is a distinguished \enquote{no-move} symbol;
		\item dually, the model can stay in its current state while only the trace proceeds one letter further which corresponds to a \emph{deletion} operation; such steps are called \emph{log moves} and can be written as $(a, \nomove)$.
	\end{enumerate}
	An alignment is complete when both, the trace and the model, have reached their final states, i.e., the model has reached its final marking and all events of the trace have been processed.
	
	\begin{figure}
		\centering
		\begin{subfigure}[b]{0.55\linewidth}
			\centering
			\begin{booktabs}{*{8}{c}}
				 $a$  &  $b$  & $\nomove$ & $\nomove$ &  $a$  &  $a$  & $\nomove$ & $\nomove$ \\ \midrule
				 $a$  &  $b$  &    $a$    &  $\tau$   &  $a$  &  $a$  &    $b$    &    $b$    \\
				$t_1$ & $t_3$ &   $t_2$   &   $t_4$   & $t_1$ & $t_2$ &   $t_3$   &   $t_5$
			\end{booktabs}
			\caption{Four synchronous moves, three model moves (inserts), one silent model move}
			\label{fig:ex:intro:alignments:a}
		\end{subfigure}
		\hfill
		\begin{subfigure}[b]{0.4\linewidth}
			\centering
			\begin{booktabs}{*{5}{c}}
				 $a$  &  $b$  &  $a$  & $\nomove$ &    $a$    \\ \midrule
				 $a$  &  $b$  &  $a$  &    $b$    &           \\
				$t_1$ & $t_3$ & $t_2$ &   $t_5$   & $\nomove$
			\end{booktabs}
			\caption{Three synchronous moves, one model move (insert), one log move (deletion)}
			\label{fig:ex:intro:alignments:b}
		\end{subfigure}
		\caption{Two possible alignments for $w=\seq{a,b,a,a}$}
		\label{fig:ex:intro:alignments}
	\end{figure}
	\zcref[S]{fig:ex:intro:alignments} shows two alignments for a trace $\sigma=\seq{a,b,a,a}$ and the system from \zcref{fig:align:example}.
	The first row indicates the progress in the trace, while the last two rows contain the labels and transitions fired by the model.
	For the first alignment, we have four insert operations (model moves) including a silent one ($t_4$ is unlabeled, indicated by $\tau$). The trace generated by the model is $\seq{a,b,a,a,a,b,b}$.
	The model trace $\seq{a,b,a,b}$ in the second alignment results in one deletion (log move) and one insertion (model move), which is optimal regarding the number of change operations.
	Thus, the costs of aligning $\sigma$ and the model would be $2$ in this example.
	
	It has frequently been observed that finding optimal alignments is computationally hard.
	Despite numerous algorithmic proposals that pushed the barrier of manageable event logs further and further, a systematic analysis of the computational structure of alignments was missing, so that surprisingly, our previous work~\autocite{SchwanenPA2025b} was the first such analysis.
	In particular, for practical applications it would be beneficial to know which parameters of the process model influence the complexity of the alignment problem and, thus, which algorithmic approaches are most promising on certain inputs.
	In the present article, we continue to overcome this situation: we show that on important classes of process models, the alignment problem is \emph{computationally harder} than on other classes (assuming $\P \neq \NP \neq \PSPACE$).
	
	In \zcref{sec:safenets}, we show that on safe Petri nets, the alignment problem is $\PSPACE$-complete and has the same complexity as the \emph{reachability problem} (cf.\@ \zcref{sec:reachability}).
	In an attempt to reduce the complexity, we then consider \emph{soundness requirements} in the form of \emph{sound workflow nets}, which are the de-facto standard in process mining.
	In \zcref{sec:wfnets}, we show that adding soundness alone does not change the picture: both, the reachability and the alignment problem, remain $\PSPACE$-complete on safe and sound workflow nets.
	Thus, we continue and impose further restrictions on the model classes.
	Specifically, we restrict the allowed \emph{choice constructs} via the notion of \emph{free-choice} Petri nets.
	In this way, we arrive at the well-studied class of \emph{live, bounded, free-choice systems} (\lbfc-systems).
	This class has particularly good algorithmic features and a well-developed structure theory. One structural property that we heavily exploit in our work is known as the \emph{Shortest Sequence Theorem}~\autocite{DeselE1995a,DeselE1995}. It gives a polynomial bound on the distance between connected markings.
	As a first step, in \zcref{sec:lbfc_systems}, we apply the structure theory of Petri nets to show that \lbfc-systems also allow for optimal alignments of polynomial length.
	This leads to new algorithmic approaches: for instance, alignments on \lbfc-systems could be computed via mixed integer linear programs~\autocite[cf.][]{SchwanenPA2025}. Compared to general safe Petri nets, this improves the complexity from $\PSPACE$ to $\NP$.
	We set out to obtain a matching lower bound for the alignment problem in the following sections.
	
	Specifically, and on the more negative side, we show that $\NP$ is the best we can hope for almost all interesting classes of Petri net systems by first drawing a connection to the \emph{membership problem} (cf.\@ \zcref{sec:membership}).
	We show in \zcref{sec:process_trees,sec:stsystems,sec:acyclic:systems} that the membership problem is $\NP$-complete for languages that can be defined by very simple process models, namely for \emph{process trees}, for \emph{T-systems} and for \emph{acyclic systems}. In particular, for the former two classes (process trees and T-systems), reachability is solvable in polynomial time, which means that the alignment problem is provably harder (assuming $\P\neq\NP$).
	Both classes have been studied extensively as models which allow no choices (T-systems) or only local choices (process trees). In~\autocite{SchwanenPA2025,SchwanenPA2025a,SchwanenPA2025c}, we already studied the complexity of alignments on process trees before. We showed that the alignment problem is in $\NP$ for general process trees and in $\P$ for process trees with \emph{unique labels}, that is, each activity occurs only once in the entire process tree.
	To complete the picture for acyclic systems, we also establish a matching $\NP$-upper bound in \zcref{sec:acyclic:systems}.
	
	To complement our findings with at least one model class with efficiently computable alignments, we finally turn our attention to \emph{live, safe S-systems}, see~\zcref{sec:stsystems}.
	These are safe systems where each transition has exactly one input and one output place and where the initial marking consists of a single token. We show that on this class, the alignment problem is solvable in polynomial time (denoted by $\P$).
	However, we further show that on S-systems, the decisive parameter is the number of tokens: the polynomial-time algorithm extends to all S-systems whose number of tokens is bounded by a constant, whereas the alignment problem becomes $\NP$-complete again as soon as the number of tokens is part of the input, both for live S-systems (with the bound given as part of the input) and for safe S-systems that are neither live nor sound.
	A complete overview of our results is given in \zcref{tab:overview:results} at the end of this article.
	
	\section{Related Work}
	
	Alignments, introduced by \autocite{Adriansyah2014,AdriansyahDA2011}, are the state-of-the-art technique and considered the \enquote{gold standard} in conformance checking~\autocite{CarmonaDSW2018,CarmonaDW2022}.
	It is ubiquitously mentioned that computing alignments (with respect to the standard algorithms based on $A^*$) suffers from the so-called \emph{state explosion problem} \autocite[cf.][]{Valmari1998}. This prevents the computation of alignments for large models or event logs. As a consequence, different algorithmic approaches were investigated to improve scalability of alignment computations. For example, in~\autocite{BloemenPA2018} the authors use symbolic representations of the (exponential) state space to reduce the memory footprint of the alignment computation. Another common approach is to encode alignments into related problems for which well-adapted algorithms exists. This idea is investigated, for example, in~\autocite{deLeoniM2017} where alignment computations are represented as a planning problem.
	Another angle is to improve heuristics for the $A^*$-algorithm: in~\autocite{vanDongen2018}, for instance, the author combines Petri net theory and linear programming to improve the runtime on several benchmarks significantly.
	Besides this, several approximative algorithms have been proposed as well, see, e.g., \autocite{TaymouriC2016} for a scheme based on integer linear programming and~\autocite{TaymouriC2018} for a genetic method to compute optimal alignments. In addition, \autocite{SchusterZA2021} presents an approximation algorithm specifically tailored to process trees.
	
	However, these studies do not provide formal guarantees and in preparation of \autocite{SchwanenPA2025b}, we did hardly find any source that investigates the algorithmic complexity of alignments. In general, it has been recognized that the reachability problem for Petri nets is a lower bound for the complexity \autocite{CarmonaDSW2018,CarmonaDW2022}.
	A statement on specific model classes is provided in \autocite{BoltenhagenCC2019,BoltenhagenCC2021}. There, the authors showed that computing alignments is $\NP$-complete on the class of safe Petri nets when the length of permissible alignments is limited.
	Beyond that, we are not aware of any results studying the computational complexity of alignments over different process models.
	
	The reachability problem for Petri nets provides a lower bound for computing alignments.
	In fact, the expressive power and algorithmic properties of formal languages defined by (general) Petri nets have already been studied in the 1970s, e.g., by \textcite{Hack1976}. Among many other results, he proved the inter-reducibility between the membership problem and the reachability problem, a construction we will use later on.
	Most relevant for us, however, are complexity bounds for classes of safe Petri nets (which means that places can hold at most one token).
	For safe Petri nets, it was shown that reachability is $\PSPACE$-complete~\autocite{JonesLL1977,ChengEP1993,ChengEP1995}.
	Later on, it turned out that almost all interesting computational problems on safe Petri nets are $\PSPACE$-hard. A comprehensive overview of results is given in~\autocite{ChengEP1993,ChengEP1995,Esparza1998,EsparzaN1994,JonesLL1977} as well as in \zcref{sec:reachability} and \autocite{SchwanenP2026}, which also include more recent results.
	Because of the high complexity, people also looked into more restricted classes of Petri nets with better algorithmic properties. Here, \emph{free-choice} Petri nets form the most relevant example on which important computational problems become tractable, see~\autocite{DeselE1995}.
	
	Finally, we like to mention work on the \emph{error correction problem} for regular languages. Here, the goal is to compute the edit distance between an input string and a regular expression. This is strongly related to the alignment problem with the most significant difference being the presentation of the regular language (a safe Petri net also defines a regular language, but in an exponentially more succinct form). A first efficient algorithm for the error correction problem for regular languages was given by \textcite{Wagner1974} in the 1970s, a complexity-theoretic analysis can be found in~\autocite{Pighizzini2001}. For context-free grammars (CFGs), a polynomial-time error-correction algorithm can be found in~\autocite{AhoP1972}. Our work might be interesting for the error correction problem as well, as it provides a new perspective on the complexity of the problem for other presentations of regular languages.
	
	\section{Preliminaries}
	
	Let $\N \coloneq \Set{0, 1, 2, \dots}$ denote the natural numbers.
	For any tuple $a$, $\pi_i(a)$ denotes the \emph{projection} on its $i$th element, i.e., $\pi_i\colon A_1\times\dots\times A_n\to A_i,(a_1,\dots,a_n)\mapsto a_i$.
	
	\begin{definition}[Multiset]
		A \emph{multiset} $M$ over a set $A$ is a function $M\colon A\to\N$; thus, for any $a\in A$, $M(a)$ indicates the multiplicity of $a$ in the multiset~$M$.
		The set of all multisets over $A$ is given by $\N^A$.
		We also use the notation $\Multiset*{a^{M(a)}\given a\in A}$ for a multiset $M\in\N^A$.
		Any set $A$ can also be considered a multiset by assigning $1$ to each item.
		For multisets $M,M'\in\N^A$, we use the standard notation for functions, e.g., $M+M'$, $M\leq M'$, etc.
		The \emph{support} (or \emph{domain}) of a multiset $M\in\N^A$, denoted by $\supp{M}$, is the set of distinct elements contained in $M$, i.e., $\supp{M}\coloneq\Set{a\in A\given M(a)>0}$.
	\end{definition}
	Whenever we compare or operate with multisets $M\in\N^A$ and $M'\in\N^B$ over different domains $A$ and $B$, we implicitly assume that they are first extended with 0 for each item of $A\cup B$ not in their domain.
	
	\begin{definition}[Sequence, Permutation]
		\emph{Sequences} with index set $I$ over a set $A$ are denoted by $\sigma=\seq{a_i}_{i\in I}\in A^I$.
		The \emph{length} of a sequence $\sigma$ is written as $\abs\sigma$ and the set of all finite sequences over $A$ is denoted by $A^*$.
		Given two sequences $\sigma$ and $\sigma'$, $\sigma\cdot\sigma'$ (or $\sigma\sigma'$ in short) denotes the concatenation of the two sequences.
		For a sequence $\sigma=\seq{a_i}_{i\in I}\in A^I$, the notation $\seqsum\sigma$ is used as a shorthand for $\sum_{i\in I}a_i$.
		The restriction of a sequence $\sigma\in A^*$ to a set $B\subseteq A$ is the subsequence $\sigma\vert_{B}$ of $\sigma$ consisting of all elements in $B$.
		A function $f\colon A\to B$ can be applied to a sequence $\sigma\in A^*$ given the recursive definition $f(\seq{})\coloneq\seq{}$ and $f(\seq a\cdot\sigma)\coloneq\seq{f(a)}\cdot f(\sigma)$.
		The \emph{Parikh vector} of a sequence $\sigma\in A^*$, denoted by $\vec{\sigma}$, is its multiset representation defined by $\vec{\sigma}(a)\coloneq\abs*{\sigma\vert_{\Set{a}}}$ for every $a\in A$, and $\supp{\vec{\sigma}}$ provides the \emph{support} of $\sigma$.
		A sequence $\sigma'\in A^*$ is a \emph{permutation} of $\sigma$ if and only if $\vec{\sigma'}=\vec{\sigma}$.
	\end{definition}
	
	\begin{definition}[Alphabet, Word, Language]
		An \emph{alphabet} $\Sigma$ is a finite, non-empty set of \emph{labels} (also referred to as \emph{activities}). A \emph{word} $w\in\Sigma^*$ is a sequence of labels and a \emph{language} $\mathcal{L}\subseteq\Sigma^*$ is a set of words.
	\end{definition}
	
	\begin{definition}[Petri Net]
		Let $\Sigma$ be an alphabet. A \emph{Petri net} $N$ is a weakly-connected bipartite directed graph $N = (P, T, F, \ell)$ where $P$ and $T$, $P \cap T = \emptyset$ are disjoint finite sets of vertices and $F \subseteq (P \times T) \cup (T \times P)$ is the set of arcs. In a Petri net, $P$ is called the set of \emph{places}, $T$ the set of \emph{transitions}, and $F$ the \emph{flow relation}. In addition, $\ell\colon T \to \Sigmatau$ is a \emph{labeling function}. A transition $t \in T$ is \emph{labeled} if $\ell(t) \in \Sigma$ and it is \emph{silent} if $\ell(t)=\tau$.
		Given a vertex $v \in P \cup T$, its \emph{pre-set} $\pset v$ and \emph{post-set} $v\pset$ are defined by $\pset v \coloneq \Set{u \in P \cup T \given (u, v) \in F}$ and $v\pset \coloneq \Set{u \in P \cup T \given (v, u) \in F}$.
		With regard to a place (transition), its pre- and post-set are also called \emph{input} and \emph{output transitions} (\emph{places}).
	\end{definition}
	
	\begin{definition}[Free-Choice Petri Net]
		A Petri net $N=(P,T,F,\ell)$ is \emph{free-choice} if any two transitions either share all or none of their input places, i.e., $\forall t,t'\in T\colon \pset t=\pset t' \text{ or } \pset t\cap\pset t'=\emptyset$.
	\end{definition}
	
	\begin{definition}[Conflict-Free Petri Net]
		A Petri net $N=(P,T,F,\ell)$ is \emph{conflict-free} if any place $p\in P$ has either at most one output transition or all output transitions are on a self-loop with that place, i.e., $\forall p\in P\colon \abs{p\pset}\geq1 \implies p\pset\subseteq\pset p$.
	\end{definition}
	
	\begin{definition}[Marking, System, Firing Rule]\label{def:marking}
		Given a Petri net $N=(P,T,F,\ell)$, a \emph{marking} $M\in \N^P$ is a multiset where $M(p)$ is the number of \emph{tokens} at place $p\in P$.
		A place $p\in P$ is \emph{marked} at $M$ if $M(p)>0$. A place that is not marked is called \emph{empty}.
		The pair $(N,M)$ of a Petri net $N=(P,T,F,\ell)$ and a marking $M\in\N^P$ is called a \emph{system}.
		A transition $t\in T$ is \emph{enabled} in $M$, denoted by $(N,M)\fire t$, if and only if each of its input places $p\in\pset t$ is marked, i.e., $\forall p\in\pset t\colon M(p)>0$. An enabled transition may \emph{fire}, denoted by $(N,M)\fire{t}(N,M')$, and firing results in a new marking $M'$:
		\begin{equation*}
			M'(p)=\begin{cases}
				M(p)-1 & \text{if } p\in\pset t\wedge p\notin t\pset, \\
				M(p)+1 & \text{if } p\notin\pset t\wedge p\in t\pset, \\
				M(p)   & \text{otherwise}.
			\end{cases}
		\end{equation*}
		
		A sequence of transitions $\sigma=\seq{t_i}_{i=1}^n\in T^*$ is called a \emph{firing sequence} of $(N,M)$ if for every transition $t_i$ of the sequence holds that $(N,M_{i-1})\fire{t_i}$ and $(N,M_{i-1})\fire{t_i}(N,M_i)$ where $M_0=M$ and $M_n=M'$. Firing such a sequence is denoted by $(N,M)\fire{\sigma}(N,M')$. The empty sequence $\seq{}$ is always enabled and firing the empty sequence leaves the marking unchanged, i.e., $(N,M)\fire{\seq{}}(N,M)$. A marking $M'$ is \emph{reachable} if a firing sequence $\sigma\in T^*$ exists such that $M'$ is the resulting marking, i.e., $(N,M)\fire{\sigma}(N,M')$. The set of all reachable markings of $(N,M)$ is denoted by $\reach{N,M}\coloneq\Set*{M'\in\N^P\given\exists\sigma\in T^*\colon(N,M)\fire{\sigma}(N,M')}$.
		A transition $t\in T$ is \emph{dead} if it is not enabled in any of the reachable markings, i.e., $\nexists M'\in\reach{N,M}\colon(N,M')\fire{t}$.
	\end{definition}
	
	\begin{definition}[Boundedness, Safeness]
		Given a $b\in\N$, a system $(N,M_0)$ with $N=(P,T,F,\ell)$ is \emph{$b$-bounded} if $b$ is a bound for any reachable marking, i.e., $\forall M\in\reach{N,M_0}\colon\forall p\in P\colon M(p)\leq b$. $(N,M_0)$ is \emph{safe} if it is $1$-bounded.
	\end{definition}
	
	\begin{definition}[(Quasi-)Liveness]
		Let $(N,M_0)$ with $N=(P,T,F,\ell)$ be a system.
		A transition $t\in T$ is \emph{quasi-live} if there exists a reachable marking $M'\in\reach{N,M_0}$ which enables $t$, i.e., $\exists M\in\reach{N,M_0}\colon(N,M)\fire{t}$.
		A transition $t\in T$ is \emph{live} if it is quasi-live for every reachable marking $M\in\reach{N,M_0}$, i.e., $\forall M\in\reach{N,M_0}\colon\exists M'\in\reach{N,M}\colon(N,M')\fire{t}$.
		A system $(N,M_0)$ is \emph{(quasi-)live} if every transition $t\in T$ is (quasi-)live.
	\end{definition}
	
	\begin{definition}[Home Marking, Cyclic System]
		Let $(N,M_0)$ be a system. A marking $M\in\reach{N,M_0}$ is a \emph{home marking} if it can be reached from any reachable marking, i.e., $\forall M'\in\reach{N,M_0}\colon M\in\reach{N,M'}$.
		A system $(N,M_0)$ is \emph{cyclic} if its initial marking $M_0$ is a home marking.
	\end{definition}
	
	\begin{definition}[Persistent System]
		A system $(N,M_0)$ with $N=(P,T,F,\ell)$ is \emph{persistent} if firing a transition does not disable any other enabled transition, i.e., $\forall t_1,t_2\in T,t_1\neq t_2\colon \forall M\in\reach{N,M_0}\colon (N,M)\fire{t_1} \wedge (N,M)\fire{t_2} \implies (N,M)\fire{t_1 t_2}$.
	\end{definition}
	
	\begin{definition}[Accepting System]
		An \emph{accepting system} is a tuple $(N, M_\init, M_\final)$ where $N = (P, T, F, \ell)$ is a Petri net, $M_\init \in \N^P$ is the \emph{initial marking}, and $M_\final \in \N^P$ is the \emph{final marking}, i.e., an accepting system extends a system with a final marking.
	\end{definition}
	
	\begin{definition}[Complete Firing Sequence, Trace, Behavior and Language of a System]
		Let $S=(N,M_\init,M_\final)$ be an accepting system with $N=(P,T,F,\ell)$ and $\ell\colon T\to\Sigmatau$. A firing sequence $\sigma\in T^*$ is a \emph{complete firing sequence} of $S$ if $(N,M_\init)\fire\sigma(N,M_\final)$. The set of complete firing sequences $\scfs{S}\coloneq\Set*{\sigma\in T^*\given(N,M_\init)\fire\sigma(N,M_\final)}$ is the \emph{behavior} of $S$.
		Considering the labels of a complete firing sequence, this is referred to as a \emph{trace} $\sigma\in\Sigma^*$. The set of all traces $\lang{S}\coloneq\Set*{\ell(\sigma)\vert_\Sigma\given\sigma\in\scfs{S}}$ is the \emph{language} of $S$.
	\end{definition}
	
	\begin{definition}[Easy Soundness]
		An accepting system $(N,M_\init,M_\final)$ is \emph{easy sound} if and only if the final marking is reachable, i.e., $M_\final\in\reach{N,M_\init}$.
	\end{definition}
	
	\section{Alignments}
	\label{sec:alignments}
	
	\emph{Alignments} \autocite{Adriansyah2014} juxtapose an observed trace with a complete firing sequence of the process model. For this, activities in the trace are compared in pairs with the transitions of the complete firing sequence.
	Every such pair either contains an activity from the trace and its corresponding transition from the firing sequence or\textemdash in case of a deviation\textemdash just one of them while the other position in the pair remains vacant, indicated by a special \emph{no-move symbol} $\nomove$. These pairs are called \emph{moves}.
	\begin{definition}[(Legal) Move]
		Let $\Sigma$ be an alphabet and $S=(N,M_\init,M_\final)$ an accepting system with Petri net $N=(P,T,F,\ell)$ and labeling function $\ell\colon T\to\Sigmatau$.
		Furthermore, let $\nomove$ be a distinguished \emph{\enquote{no move} symbol}.
		A \emph{move} is an ordered pair $(a,t)\in\left(\Sigma\cup\Set\nomove\right)\times\left(T\cup\Set\nomove\right)$. We distinguish between three types of \emph{legal} moves: The move $(a,t)$ is a
		\begin{itemize}
			\item \makebox[10em][l]{\emph{synchronous move}} if  $a\in\Sigma$, $t\in T$, and $a=\ell(t)$,
			\item \makebox[10em][l]{\emph{log move}} if $a\in\Sigma$ and $t=\nomove$,
			\item \makebox[10em][l]{\emph{model move}} if $a=\nomove$ and $t\in T$.
		\end{itemize}
		All other moves are considered illegal.
		A model move $(\nomove,t)$ is a \emph{silent move} if $\ell(t)=\tau$.
		The set of all legal moves between $\Sigma$ and $S$ is denoted by
		\begin{equation*}
			\mathit{LM_{\mathrm\Sigma,S}}\coloneq\Set*{(a,t)\in\Sigma\times T\given a=\ell(t)}\cup(\Sigma\times\Set{\nomove})\cup(\Set{\nomove}\times T).
		\end{equation*}
	\end{definition}
	An \emph{alignment} is a sequence of legal moves whose first components form the observed trace and whose second components form a complete firing sequence of the process model (ignoring the $\nomove$-symbol and $\tau$-labels), formally:
	\begin{definition}[Alignment]\label{def:alignment}
		Let $\Sigma$ be an alphabet, $\sigma\in\Sigma^*$ a trace, and $S=(N,M_\init,M_\final)$ an accepting system with $N=(P,T,F,\ell)$ and $\ell\colon T\to\Sigmatau$.
		An \emph{alignment} $\gamma\in\mathit{LM_{\mathrm\Sigma,S}}^*$ between $\sigma$ and $S$ is a sequence of moves such that $\sigma=\pi_1(\gamma)\vert_\Sigma$ and $\pi_2(\gamma)\vert_T\in\scfs{S}$. $\Gamma_{\sigma,S}$ denotes all alignments between $\sigma$ and $S$:
		\begin{equation*}
			\Gamma_{\sigma,S}\coloneq\Set*{\gamma\in\mathit{LM_{\mathrm\Sigma,S}}^*\given\pi_1(\gamma)\vert_\Sigma=\sigma\wedge\pi_2(\gamma)\vert_T\in\scfs{S}}.
		\end{equation*}
	\end{definition}
	
	In easy-sound systems, a trivial alignment always exists: first generate the input trace via log moves, and then generate a firing sequence from the initial to the final marking via model moves.
	This trivial alignment corresponds to the worst possible scenario: the model and the trace have nothing in common. 
	However, we are really interested in \emph{optimal} alignments which \emph{maximize} the synchronization between the trace and the model, or, to put it differently, which \emph{minimize} the deviations.
	This is formalized by assigning costs to moves and then finding an alignment with minimal costs.
	Here, the costs of an alignment are determined as the sum of costs for each move in the sequence.
	\begin{definition}[Alignment Cost, Standard Cost Function]
		Let $S=(N,M_\init,M_\final)$ be an easy-sound system, i.e., $\scfs{S}\neq\emptyset$, with $N=(P,T,F,\ell)$ and $\ell\colon T\to\Sigmatau$, and let $\mathit{LM_{\mathrm\Sigma,S}}$ be the set of all legal moves between $\Sigma$ and $S$. An \emph{alignment cost function} is a function $c\colon\mathit{LM_{\mathrm\Sigma,S}}\to\costRange$. The cost of an alignment $\gamma\in\mathit{LM_{\mathrm\Sigma,S}}^*$ is given by the sum of costs of each move in the sequence, i.e., $\seqsum c(\gamma)$.
		
		The \emph{standard cost function} $c\colon\mathit{LM_{\mathrm\Sigma,S}}\to\costRange$ is defined as
		\begin{equation*}
			(a,t)\mapsto c(a,t)\coloneq\left\lbrace\begin{array}{lllll}
				0 & \quad & a\in\Sigma\wedge t\in T\wedge a=\ell(t), &\text{ or } & a=\nomove\wedge \ell(t)=\tau, \\
				1 & & a\in\Sigma\wedge t=\nomove, &\text{ or } & a=\nomove\wedge \ell(t)\in\Sigma.
			\end{array}\right.
		\end{equation*}
	\end{definition}
	An alignment is considered to be \emph{optimal} if and only if there is no other alignment with lower costs. Thus, it aligns a complete firing sequence of the process model that most closely matches the observed trace according to the cost function.
	In addition to the standard cost function, these include, for example, functions that penalize only log moves or define certain activities in the model as milestones, meaning that their omission (i.e., involvement in a model move) incurs an excessive penalty~\autocite{BloemenZADP2022}.
	We are mostly interested in \emph{optimal} alignments that match the observed trace as closely as possible (according to a cost function):
	\begin{definition}[Optimal Alignment]
		Let $S$ be an easy-sound system, i.e., $\scfs{S}\neq\emptyset$, and let $c\colon\mathit{LM_{\mathrm\Sigma,S}}\to\costRange$ be an alignment cost function. Given a trace $\sigma\in\Sigma^*$, an alignment $\gamma_{\mathit{opt}}\in\Gamma_{\sigma,S}$ is \emph{optimal} if and only if no other alignment between $\sigma$ and $S$ has lower costs, i.e.,
		\begin{equation*}
			\seqsum c(\gamma_{\mathit{opt}})=\min_{\gamma\in\Gamma_{\sigma,S}}\Set*{\seqsum c(\gamma)}.
		\end{equation*}
	\end{definition}
	Note that the minimum always exists, since costs of (legal) moves are non-negative, which implies that we cannot decrease costs of an alignment by adding extra moves. Also, for each trace $\sigma \in \Sigma^*$ and system $S$ with $\scfs{S} \neq \emptyset$, there exists at least one trivial alignment: first use moves of the form $(a, \nomove)$, for $a \in \Sigma$, in to generate the input trace $\sigma$, and then we take any firing sequence $\sigma' \in \scfs{S}$ to let the system reach its final marking via moves of the form $(\nomove, t)$, for $t \in T$. Hence, in easy-sound systems we get a trivial upper bound on the costs of a minimal alignment.
	
	While computing an optimal alignments is a \emph{functional} optimization problem, the corresponding decision problem is the following:
	\begin{problem}{Alignment $(\CPalign)$}
		\emph{Input:} An alphabet $\Sigma$, an easy-sound system $S=(N,M_\init,M_\final)$, $\scfs{S}\neq\emptyset$, with Petri net $N=(P,T,F,\ell)$ and labeling function $\ell\colon T\to\Sigmatau$, a trace $\sigma\in\Sigma^*$ over $\Sigma$, and a cost function $c\colon\mathit{LM_{\mathrm\Sigma,S}}\to\costRange$.
		
		\emph{Question:} Given a threshold $k\in\costRange$, is there an alignment $\gamma\in\Gamma_{\sigma,S}$ with $\seqsum c(\gamma)\leq k$?
	\end{problem}
	For $\CPalign$ and all decision problems presented throughout this article, we use standard encoding, i.e., all numbers including the values of the cost function and the marking are represented in binary notation. The only exception is the bound for bounded Petri nets, which we assume to be in unary notation.
	
	Although $\CPalign$ is a \emph{functional} optimization problem (for which we are interested in the optimal value as output),
	for our purposes the (binary) decision problem $\CPalign$ is more adequate since we are interested in classifying its computational complexity.
	However, it is also clear that an efficient algorithm for the decision version can be transformed into an efficient algorithm for the functional variant by performing a binary search on the cost threshold $k$. This only requires a polynomial number of calls to the decision algorithm and thus does not change the complexity class of the problem.
	
	The complexity in computing alignments comes from the huge number of possible firing sequences that are candidates for an optimal solution. Of course, testing all sequences one after the other is not a viable strategy as the search space is way too large. Instead, the best-suited firing sequence of the process model has to be found by different means.
	Even if the explicit generation of the full language of the process model proves to be difficult, it is implicitly encoded in the Petri net. This can be used to determine an optimal alignment by connecting trace and model in such a way that directly generates alignments.
	
	\subsection{Petri Nets as Language Acceptors, Workflow Nets, and the Soundness Property}
	\label{sec:alignments:wfnets}
	
	While Petri nets are frequently considered in combination with an initial marking, one rarely assumes a final state as well. However, this changes when Petri net languages are studied, where, depending on the type, a final marking naturally comes into play. In the study of Petri net languages various types are distinguished, such as terminal languages (labeled sequences that end in a final marking), covering or weak languages (labeled sequences that end in a marking covering the final marking), prefix-closed languages (all labeled sequences), and dead languages (labeled sequences that end in a deadlock) (see, e.g., \autocite{Hack1976,Jantzen1986}). Another distinction is made based on the labeling function, e.g., if one allows for silent transitions or if it is restricted to unique labels (i.e., each label can be used by only one transition).
	
	Alignments determine the closest trace of a model to a given one, that is, we usually consider \emph{full} executions, i.e., the terminal language which consists of traces from the initial state to the final state of the model. In case of partial observations and thus only a part of a trace is given, the computation of alignments can be adapted to prefixes, postfixes, or infixes where deviations at the beginning or end of an alignment are not penalized \autocite{Adriansyah2014}. Nevertheless, the model trace still needs to start from the initial marking and end at the final marking. Consequently, in the context of process mining and alignments in particular, we normally have some final or accepting state. In the following, we bound our scope to standard alignments as already defined above. Indeed, business processes are usually described by modeling the behavior of a single case with a clear start and termination point. Therefore, \autocite{vanderAalst1997} introduced workflow nets, which have evolved into the de-facto modeling standard in process mining.
	\begin{definition}[Workflow Net]
		\label{def:wfnet}
		A \emph{workflow net} $N=(P,T,F,\ell,i,o)$ is a Petri net $(P,T,F,\ell)$ where $i,o\in P$ are special places, namely the single \emph{source place} $i$ and the single \emph{sink place} $o$ with $o\pset=\emptyset$, and the Petri net is strongly connected if a transition $\bar{t}$ with $\pset\bar{t}=\Set{o}$ and $\bar{t}\pset=\Set{i}$ is added, i.e., every vertex $v\in P\cup T$ of the Petri net is on a path from $i$ to $o$.
		Implicitly, in workflow nets, the initial marking is $\Multiset{i}$ and the final marking is $\Multiset{o}$.
	\end{definition}
	Please note that we deviate slightly from the classic definition of workflow nets. Commonly, the source place of a workflow net is required to not have any incoming arcs, i.e., $\pset i=\emptyset$. However, this restriction can be waived as any limitation with regard to the source is implied by the initial marking~$\Multiset{i}$ rather than the fact that the source has no incoming arcs.
	
	Due to their motivation, workflow nets are acceptors by design. Consequently, this raises the question of whether their structural design allows them to always reach the accepting state. Therefore, the notion of \emph{soundness} was introduced at the same time. Soundness describes the \enquote{well-behavedness} of a workflow net as business procedures should terminate eventually with leaving no tokens behind and each task in the procedure (i.e., transition) should potentially be executable~\autocite{vanderAalst1997}.
	\begin{definition}[Soundness of Workflow Nets]
		A workflow net $N=(P,T,F,\ell,i,o)$ is \emph{sound} if and only if it satisfies the following requirements:
		\begin{itemize}
			\item \makebox[10em][l]{option to complete,} i.e., $\forall M\in\reach{N,\Multiset{i}}\colon\Multiset{o}\in\reach{N,M}$,
			\item \makebox[10em][l]{proper completion,} i.e., $\forall M\in\reach{N,\Multiset{i}}\colon M\geq\Multiset{o}\implies M=\Multiset{o}$, and
			\item \makebox[10em][l]{no dead transitions,} i.e., $\forall t\in T\colon\exists M\in\reach{N,\Multiset{i}}\colon(N,M)\fire{t}$.
		\end{itemize}
	\end{definition}
	In this definition, the option to complete also already implies proper completion~\autocite{vanderAalstHHSVVW2011} and having no dead transitions is the same as quasi-liveness. This is because there are no dominating markings in workflow nets with the option to complete~\autocite{SchwanenP2026}.
	However, the notion of soundness is not limited to workflow nets, but can be applied to all accepting systems.
	\begin{definition}[Soundness of Accepting Systems]
		\label{def:soundness}%
		An accepting system $(N,M_\init,M_\final)$ is \emph{sound} if and only if it satisfies the following requirements:
		\begin{itemize}
			\item \makebox[10em][l]{option to complete,} i.e., $\forall M\in\reach{N,M_\init}\colon M_\final\in\reach{N,M}$,
			\item \makebox[10em][l]{proper completion,} i.e., $\forall M\in\reach{N,M_\init}\colon M\geq M_\final\implies M=M_\final$, and
			\item \makebox[10em][l]{quasi-liveness,} i.e., $\forall t\in T\colon\exists M\in\reach{N,M_\init}\colon(N,M)\fire{t}$.
		\end{itemize}
	\end{definition}
	Be aware that quasi-liveness and the option to complete alone do not suffice in this more general setting as proper completion is not automatically implied anymore. This implication goes even further: while for workflow nets also boundedness directly follows from soundness, this no longer applies to general accepting systems.
	
	Soundness guarantees that the accepting marking is reachable from every other reachable marking. This is not necessary for alignments: they only require easy soundness, i.e., that the accepting marking is reachable from the initial marking. Basically, the language accepted by the system should not be empty. In fact, the easy-soundness requirement is merely a technicality: although alignments are not defined on an empty language, the main reason is that the $A^*$-based techniques to compute alignments require reachability of the final marking in order to guarantee termination~\autocite{CarmonaDSW2018}. A notable observation regarding alignments related to soundness is worth mentioning: because alignments must necessarily reach the final marking, they effectively reduce their consideration of the accepting system to the part of it that is actually sound. Thus, in the context of other types of Petri net languages, it might also be interesting to adapt alignments accordingly and investigate the effect of different soundness notions (see, e.g., \autocite{vanderAalstHHSVVW2011}).
	
	Workflow nets are clearly accepting systems, but any proper accepting system (i.e., an accepting system where each transition has at least one input place and one output place) can also be converted into a workflow net quite easily. To do this, simply add a new source from which the original initial marking is generated via a silent start transition, and also add an end transition that consumes the original final marking and puts a token in a new sink. But of course, some Petri net properties might be lost throughout this conversion. Therefore, we here extended our scope beyond workflow nets. However, there is one mild condition that we impose on accepting systems, namely that their underlying graph must be weakly connected. As we will see, concurrency plays a key role in the complexity of alignments. Hence, with Petri nets being weakly-connected, we exclude the trivial case of concurrency caused by the parallel execution of independent nets. Note, however, that weak connectivity can also be achieved via dead transitions, so it does not exclude such concurrency in general (cf.\@ \zcref{thm:align_safe_ssystem_np}); for live or sound systems, this is different (cf.\@ \zcref{lem:safe_ssystem_single_token}).
	
	Let us make another important remark which is valid throughout the rest of this article: whenever we consider computational problems involving \emph{bounded} Petri nets, we always assume that the bound $b\in\N$ is part of the input \emph{and} that the bound $b$ is represented in unary notation.
	Our motivation for this requirement is as follows.
	It is well-known that, in general, the bound $b\in\N$ of a \emph{bounded} Petri net can be extremely large, even super exponential in the size of the Petri net.
	Since we are, however, interested in low complexity classes (such as $\P$, $\NP$ and $\PSPACE$), a super-exponential bound would be rather meaningless, as we could not even represent it using polynomial space.
	By requiring that the bound is given in unary notation as part of the input, we ensure that a runtime which is polynomial in the bound $b$ still accounts for an \emph{efficient} runtime.
	This assumption also goes in line with the fact that in almost all works from the process mining area, a \emph{bounded} Petri net system usually has a small constant bound (most frequently $b=1$ which gives us the setting of safe systems).
	We leave it as an interesting academic endeavor for the future to investigate in how far our complexity results change if we drop this assumption.
	
	\subsection{Synchronous Product}
	\label{sec:alignments:optimal}
	
	The standard approach for computing optimal alignments is via a reduction to a reachability problem in Petri nets.
	The first step is to express the trace itself as a Petri net:
	\begin{definition}[Trace System]\label{def:trace_system}
		Let $\Sigma$ be an alphabet and let $\sigma\in\Sigma^*$ be a trace over $\Sigma$. The \emph{trace system} of trace $\sigma$, denoted by $\mathcal T(\sigma)$, is an accepting system $\mathcal T(\sigma)\coloneq(N,M_\init,M_\final)$ with a Petri net $N=(P,T,F,\ell)$ where
		\begin{itemize}
			\item $P\coloneq\Set{p_i\given 0\leq i\leq\abs\sigma}$ is the set of places,
			\item $T\coloneq\Set{t_i\given 1\leq i\leq\abs\sigma}$ is the set of transitions,
			\item $F\coloneq\bigcup_{1\leq i\leq\abs\sigma}\Set{(p_{i-1},t_i),(t_i,p_i)}$ is the flow relation,
			\item $\ell\colon T\to\Sigma$ is the labeling function such that $\ell\big(\seq{t_i}_{i=1}^{\abs\sigma}\big)=\sigma$,
			\item $M_\init\coloneq\Multiset{p_0}$ is the initial marking, and
			\item $M_\final\coloneq\Multiset{p_{\abs\sigma}}$ is the final marking.
		\end{itemize}
	\end{definition}
	
	Now, both, the process model and the trace are represented by a Petri net and can be combined using the synchronous product, which was introduced in \autocite{Adriansyah2014} and is a special case of the product of Petri nets introduced in \autocite{Winskel1987}.
	\begin{definition}[Synchronous Product]\label{def:synchronous_product}
		Let $\Sigma$ be an alphabet and let $S_1=(N_1,M_{1,\init},M_{1,\final})$ and $S_2=(N_2,M_{2,\init},M_{2,\final})$ be two accepting systems with Petri nets $N_1=(P_1,T_1,F_1,\ell_1)$ and $N_2=(P_2,T_2,F_2,\ell_2)$ and labeling functions $\ell_1\colon T_1\to\Sigmatau$ and $\ell_2\colon T_2\to\Sigmatau$ where $P_1$, $T_1$, $P_2$, and $T_2$ are pairwise disjoint sets. Furthermore, let $\nomove\notin\Sigma,T_1,T_2$ be a distinguished symbol. The \emph{synchronous product} of $S_1$ and $S_2$, denoted as $S_1\otimes S_2$, is the accepting system $S_1\otimes S_2\coloneq(N,M_\init,M_\final)$ with the Petri net $N\coloneq(P,T,F,\ell)$ and labeling function $\ell\colon T\to(\Sigma\cup\Set{\tau,\nomove})\times(\Sigma\cup\Set{\tau,\nomove})$ where
		\begin{itemize}
			\item $P\coloneq P_1\cup P_2$,
			\item $T\coloneq\Set{(t_1,t_2)\in T_1\times T_2\given \ell_1(t_1)=\ell_2(t_2)}\cup(T_1\times\Set\nomove)\cup(\Set\nomove\times T_2)$,
			\item $F\coloneq\begin{aligned}[t]
				&\Set*{(p,(t_1,t_2))\in P\times T\given (p,t_1)\in F_1 \vee (p,t_2)\in F_2}\\
				\cup\,&\Set*{((t_1,t_2),p)\in T\times P\given (t_1,p)\in F_1 \vee (t_2,p)\in F_2},
			\end{aligned}$
			\item $(t_1,t_2)\mapsto \ell(t_1,t_2)\coloneq\begin{cases}
				(\ell_1(t_1),\ell_2(t_2)) & t_1\in T_1 \wedge t_2\in T_2, \\
				(\ell_1(t_1),\nomove)     & t_2\notin T_2,                \\
				(\nomove,\ell_2(t_2))     & t_1\notin T_1,
			\end{cases}$
			\item $M_\init\coloneq M_{1,\init}+M_{2,\init}$, and $M_\final\coloneq M_{1,\final}+M_{2,\final}$.
		\end{itemize}
	\end{definition}
	As shown in \autocite{Adriansyah2014}, complete firing sequences in the synchronous product correspond to alignments between the trace and the model.
	\begin{proposition}[{{\normalfont\autocite[Theorem~4.3.5]{Adriansyah2014}}}]
		Given a trace $\sigma$ and an accepting system $S$ as process model, complete firing sequences of their synchronous product correspond to alignments between $\sigma$ and $S$, i.e., $\Gamma_{\sigma,S}=\lang{\mathcal{T}(\sigma)\otimes S}$.
	\end{proposition}
	
	Furthermore, there are some additional properties of the synchronous product which will prove useful for this work.
	\begin{proposition}[{{\normalfont\autocites[Theorem~4.3.4]{Adriansyah2014}[Theorem~4.1]{Winskel1987}}}]
		Given two systems $S_1$ and $S_2$, any combination $M_1+M_2$ of a reachable marking $M_1\in\reach{S_1}$ and a reachable marking $M_2\in\reach{S_2}$ is a reachable marking in the synchronous product $S_1\otimes S_2$, i.e., $\forall M_1\in\reach{S_1},M_2\in\reach{S_2}\colon M_1+M_2\in\reach{S_1\otimes S_2}$.
	\end{proposition}
	\begin{corollary}\label{cor:synchronous_reachability_set}
		Given two systems $S_1$ and $S_2$, the set of reachable markings of their synchronous product $S_1\otimes S_2$ is the product of the sets of reachable markings of the systems, i.e., $\reach{S_1\otimes S_2}=\reach{S_1}\times\reach{S_2}$.
	\end{corollary}
	\begin{corollary}\label{cor:synchronous_product_boundedness}
		Given a $b_1$-bounded system $S_1$ and a $b_2$-bounded system $S_2$, their synchronous product $S_1\otimes S_2$ is $\max\Set{b_1,b_2}$-bounded.
	\end{corollary}
	
	\section{Reachability as a Lower Bound for Alignments}
	\label{sec:reachability}
	
	The reachability problem for Petri nets is generally considered a lower bound for the complexity of the alignment problem \autocite{CarmonaDSW2018,CarmonaDW2022}.
	In the following, we not only give a proof for this relationship, but we also draw a connection between known complexity results for the reachability problem on various Petri net classes and their implications on the alignment problem.
	
	The reachability problem asks whether two given markings are reachable from one another and it is one of the classic algorithmic benchmarks in Petri net theory. Other decision problems for Petri nets are often  equivalent to the reachability problem, manifesting its central role when analyzing Petri nets.
	\begin{problem}{Reachability ($\CPreach$)}
		\emph{Input:} A system $(N,M_0)$ with $N=(P,T,F,\ell)$ and a marking $M\in\N^P$.
		
		\emph{Question:} Is $M\in\reach{N,M_0}$?
	\end{problem}
	
	\begin{table}
		\newcommand{\rotc}[1]{\rotatebox[origin=c]{90}{#1}}
		\newcommand*{\PNprop}{$\bullet$}
		\newcommand*{\PNimpl}{$\circ$}
		\centering
		\caption{Overview of complexity results for the reachability problem on different classes of Petri nets. Assumed properties of the class considered are marked by \PNprop, thereby implied properties are marked by \PNimpl.}
		\label{tab:complexity_of_reachability}
		\begin{booktabs}{colspec={*{9}{c}l},cell{1}{1}={c=8}{c},cell{1}{9}={r=2,c=2}{c}}
			\toprule
			Petri Net Properties &                 &             &               &                   &                    &                &                       & {Complexity of the\\ Reachability Problem} &                                       \\ \cmidrule[lr]{1-8}
			   \ \rotc{live}     & \rotc{bounded*} & \rotc{safe} & \rotc{cyclic} & \rotc{persistent} & \rotc{free-choice} & \rotc{acyclic} & \rotc{conflict-free}\ &                                            &                                       \\ \midrule
			                     &                 &             &               &                   &                    &                &                       &           $\Ackermann$-complete            & \autocite{CzerwinskiO2022,Leroux2022} \\
			                     &                 &             &               &                   &      \PNprop       &                &                       &           $\Ackermann$-complete            & \autocite{CzerwinskiO2022,Leroux2022} \\ \addlinespace
			                     &                 &             &    \PNprop    &                   &                    &                &                       &            $\EXPSPACE$-complete            & \autocite{CardozaLM1976}              \\ \addlinespace
			                     &     \PNprop     &             &               &                   &                    &                &                       &             $\PSPACE$-complete             & \autocite{JonesLL1977}                \\
			                     &     \PNimpl     &   \PNprop   &               &                   &                    &                &                       &             $\PSPACE$-complete             & \autocite{ChengEP1993,ChengEP1995}    \\
			                     &     \PNimpl     &   \PNprop   &               &                   &      \PNprop       &                &                       &             $\PSPACE$-complete             & \autocite{ChengEP1993,ChengEP1995}    \\
			      \PNprop        &     \PNimpl     &   \PNprop   &    \PNprop    &                   &                    &                &                       &             $\PSPACE$-complete             & \autocite{SchwanenP2026}              \\ \addlinespace
			      \PNprop        &     \PNprop     &             &               &                   &      \PNprop       &                &                       &               $\NP$-complete               & \autocite{Esparza1998a}               \\
			      \PNprop        &     \PNimpl     &   \PNprop   &               &                   &      \PNprop       &                &                       &               $\NP$-complete               & \autocite{Esparza1998a}               \\
			                     &                 &             &               &                   &                    &    \PNprop     &                       &               $\NP$-complete               & \autocite{ChengEP1993,ChengEP1995}    \\
			                     &     \PNimpl     &   \PNprop   &               &                   &                    &    \PNprop     &                       &               $\NP$-complete               & \autocite{Stewart1995}                \\
			                     &     \PNimpl     &   \PNprop   &               &                   &      \PNprop       &    \PNprop     &                       &               $\NP$-complete               &                                       \\
			                     &                 &             &               &      \PNimpl      &                    &                &        \PNprop        &               $\NP$-complete               & \autocite{HowellR1987,HowellR1988}    \\ \addlinespace
			                     &     \PNprop     &             &               &      \PNimpl      &                    &                &        \PNprop        &                  $\in\P$                   & \autocite{HowellR1989}                \\
			                     &     \PNimpl     &   \PNprop   &               &      \PNimpl      &                    &                &        \PNprop        &                  $\in\P$                   & \autocite{HowellR1989}                \\
			      \PNprop        &     \PNprop     &             &    \PNprop    &                   &      \PNprop       &                &                       &                  $\in\P$                   & \autocite{DeselE1991,DeselE1993}      \\
			      \PNprop        &     \PNimpl     &   \PNprop   &    \PNprop    &                   &      \PNprop       &                &                       &                  $\in\P$                   & \autocite{DeselE1991,DeselE1993}      \\ \bottomrule
		\end{booktabs}
		\medskip
		
		* Holds for a $b$-bounded Petri net where $b$ is part of the input (given in unary notation).
	\end{table}
	The reachability problem is well-studied and its complexity has been investigated for many interesting subclasses of Petri nets. Although some of this complexity results are already known since the 1970s, the question on the complexity of $\CPreach$ on general Petri nets remained open for a long time. It was shown early to be $\EXPSPACE$-hard~\autocite{Lipton1976} but decidable~\autocite{Mayr1981}. However, only recent advances strengthened these bounds further by first giving an $\Ackermann$ upper bound \autocite{LerouxS2015} and then proving it to be non-elementary~\autocite{CzerwinskiLLLM2019,CzerwinskiLLLM2020}---until in 2022, when its complexity was finally answered to be $\Ackermann$-complete~\autocite{CzerwinskiO2022,Leroux2022}.
	
	In \zcref{tab:complexity_of_reachability}, we collated these recent results with those for $\CPreach$ on relevant Petri net subclasses based on overviews given in \autocite{EsparzaN1994,ChengEP1993,ChengEP1995,JonesLL1977,Esparza1998,SchwanenP2026}. Well-known Petri net classes not specifically addressed in \zcref{tab:complexity_of_reachability} are S-systems and T-systems (although T-nets are a strict subclass of conflict-free Petri nets), for both of which reachability is known to belong to $\P$~\autocite{BestT1987,CommonerHEP1971,GenrichL1973}. Furthermore, every non-free-choice Petri net can be transformed into an equivalent free-choice Petri net using a folklore approach that goes back to \autocite{Hack1974,JonesLL1977}. For each arc that conflicts with the free-choice property, the construction introduces a pair of a silent transition and a place. Although this transformation generally violates behavioral properties like liveness and cyclicity, it at least preserves boundedness. From this, we can derive that the reachability problem on the class of acyclic free-choice Petri nets (regardless of their boundedness) remains $\NP$-complete.
	
	\begin{table}
		\newcommand*{\WFprop}{$\bullet$}
		\newcommand*{\WFimpl}{$\circ$}
		\centering
		\caption{Overview of complexity results for the reachability problem on different classes of workflow nets. Structural and behavioral properties of the considered workflow net classes are: quasi-liveness (QL), the option to complete (OC), free-choiceness (FC), and safeness (safe). Assumed properties are marked by \WFprop, thereby implied properties are marked by \WFimpl.}
		\label{tab:complexity_of_reachability_wfnets}
		\begin{booktabs}{colspec={*{5}{c}l},cell{1}{1}={c=4}{c},cell{1}{5}={r=2,c=2}{c}}
			\toprule
			Workflow Net Properties &                     &                     &         & {Complexity of the\\ Reachability Problem} &                                       \\ \cmidrule[lr]{1-4}
			  \enspace QL\enspace   & \enspace OC\enspace & \enspace FC\enspace &  safe   &                                            &                                       \\ \midrule
			                        &                     &                     &         &           $\Ackermann$-complete            & \autocite{CzerwinskiO2022,Leroux2022} \\
			        \WFprop         &                     &       \WFprop       &         &           $\Ackermann$-complete            & \autocite{SchwanenP2026}              \\ \addlinespace
			                        &       \WFprop       &                     &         &               $\in\EXPSPACE$               & \autocite{SchwanenP2026}              \\
			        \WFprop         &       \WFprop       &                     &         &               $\PSPACE$-hard               & \autocite{SchwanenP2026}              \\ \addlinespace
			                        &                     &                     & \WFprop &             $\PSPACE$-complete             & \autocite{ChengEP1993,ChengEP1995}    \\
			        \WFprop         &       \WFprop       &                     & \WFprop &             $\PSPACE$-complete             & \autocite{SchwanenP2026}              \\
			        \WFprop         &                     &       \WFprop       & \WFprop &             $\PSPACE$-complete             & \autocite{SchwanenP2026}              \\ \addlinespace
			        \WFimpl         &       \WFprop       &       \WFprop       & \WFimpl &                  $\in\P$                   & \autocite{SchwanenP2026}              \\ \bottomrule
		\end{booktabs}
	\end{table}
	As already mentioned above, workflow nets constitute the standard model class in process mining. As they provide a clear start and termination point, the connection to also consider their soundness follows almost naturally.
	In light of workflow nets, we are therefore also interested in the properties related to soundness.
	\zcref[S]{tab:complexity_of_reachability_wfnets}---as already shown in~\autocite{SchwanenP2026}---highlights the complexity results for $\CPreach$ on workflow nets considering different degrees of the soundness property fulfilled, namely quasi-liveness (QL) and the option to complete (OC), combined with the effect of the free-choice property (FC). Because short-circuiting a sound workflow net does not change its set of reachable markings and yields a cyclic, live, and bounded system (\autocites[Theorem~11]{vanderAalst1997}[Lemma~1]{SchwanenP2026}), the reachability problem for cyclic, live, and bounded (even safe) systems is also $\PSPACE$-complete.
	
	Most of the existing approaches for computing optimal alignments solve a shortest path problem on the reachability graph of the synchronous product where the distance of an arc is associated with the cost of the corresponding move. This leads to an extension of the reachability problem where costs are associated to the transitions of the Petri net. By considering minimal costs, we obtain a corresponding decision problem:
	\begin{problem}{Minimum-Cost Reachability ($\CPcostReach$)}
		\emph{Input:} A system $(N,M_0)$ with $N=(P,T,F,\ell)$, a cost function $c\colon T\to\costRange$, and a marking $M\in\N^P$.
		
		\emph{Question:} Given a threshold $k\in\costRange$, is there a $\sigma\in T^*$ such that $(N,M_0)\fire\sigma(N,M)$ and $\seqsum c(\sigma)\leq k$?
	\end{problem}
	
	Since we can also assign costs of 0 to any transition, $\CPcostReach$ is a generalization of $\CPreach$.
	\begin{lemma}
		\label{lem:reach_to_mincostreach}%
		$\CPreach$ is polynomial-time reducible to $\CPcostReach$.
	\end{lemma}
	\begin{proof}
		Let $(N,M_0)$, $N=(P,T,F,\ell)$, and $M$ be an input of the $\CPreach$ problem. That is, $(N,M_0)$ is a system with a Petri net $N=(P,T,F,\ell)$, and $M\in\N^P$ is a marking.
		Let $c\colon T\to\costRange,t\mapsto c(t)\coloneq 0$ be a cost function. Then, a solution to $\CPcostReach$ with a system $(N,M_0)$ where $N=(P,T,F,\ell)$, the marking $M$ and any number $k\in\costRange$ as input is also a solution to $\CPreach$.
	\end{proof}
	
	With the extension of $\CPreach$ to $\CPcostReach$, we can now show the close connection to $\CPalign$:
	\begin{lemma}
		\label{lem:mincostreach_to_align}%
		On the class of safe systems, $\CPcostReach$ is polynomial-time reducible to $\CPalign$.
	\end{lemma}
	\begin{proof}
		Let $(N,M_0)$, $M$, $c$, and $k$ be an input for the $\CPcostReach$ problem, i.e., $(N,M_0)$ is a safe system with Petri net $N=(P,T,F,\ell)$, $M\in\N^P$ the target marking, $c\colon T\to\costRange$ a cost function, and $k\in\costRange$ a threshold.
		It is to decide if $M$ can be reached from $M_0$ with costs at most $k$.
		
		To map this to an input of $\CPalign$, we make use of the empty trace $\sigma=\seq{}$. Aligning the empty trace corresponds to finding a firing sequence from the initial marking $M_0$ to the final marking $M$ with minimal costs.
		This is because the alignment can only consist of model moves. For these model moves we can simply use the costs for firing the underlying transitions according to $c$.
		There is, however, a small issue: the system $(N,M_0,M)$ is not necessarily easy-sound. This is, however, required for inputs of $\CPalign$.
		To solve this, we add a new transition $t_\Lskip$ which allows the system to move from $M_0$ to $M$ in one step. Moreover, we make this transition very expensive (at least $k+1$). This ensures easy-soundness, and, in case $M$ is not reachable from $M_0$, the optimal alignment has costs at least $k+1$.
	\end{proof}
	Although adding the gadget of a skip transition $t_\Lskip$ is just a technicality, it nevertheless restricts the applicability of our reduction to safe systems, which is sufficient in the context of this article.
	Yet, if other model classes or additional Petri net properties are to be considered, it has to be ensured that these are not violated by adding the gadget or by replacing it with a suitable construction.
	
	By combining the two previous results, we can now reduce $\CPreach$ to $\CPalign$ in polynomial time.
	\begin{corollary}
		On the class of safe systems, $\CPreach$ is polynomial-time reducible to $\CPalign$.
	\end{corollary}
	
	Let us take the opportunity to provide two adaptions of the gadget used in the proof of \zcref{lem:mincostreach_to_align} explicitly, namely for the classes of bounded and of general systems.
	Let $M_0,M\in\N^P$ be the input markings of $\CPcostReach$ and let $d$ be the maximal token difference of a place in these markings, i.e., $d\coloneq\max_{p\in P}\abs{M(p)-M_0(p)}$. We construct the gadget in the following way: first, we add $d+1$ places $P_G$. Second, we add input transitions $T_{G,\mathit{in}}$ to the places such that each place $p\in\supp{M_0}$ marked at $M_0$ is an input place of at least $M_0(p)-M(p)$ transitions in $T_{G,\mathit{in}}$, each place in $P_G$ has exactly one input transition from $T_{G,\mathit{in}}$, and each transition in $T_{G,\mathit{in}}$ has at least one input and one output place. Third, we add output transitions $T_{G,\mathit{out}}$ such that each place in $P_G$ has exactly one output transition from $T_{G,\mathit{out}}$, each place $p\in\supp{M}$ marked at $M$ is an output place of at least $M(p)-M_0(p)$ transitions in $T_{G,\mathit{out}}$, each transition in $T_{G,\mathit{out}}$ has at least one input and one output place, and for every place $p\in\supp{M_0}\cup\supp{M}$ holds that $M_0(p)-\abs{p\pset\cap T_{G,\mathit{in}}}+\abs{\pset p\cap T_{G,\mathit{out}}}=M(p)$. Finally, we assign costs of at least $(k+1)/2$ to all added transitions in $T_{G,\mathit{in}}\cup T_{G,\mathit{out}}$.
	The resulting gadget allows us to transition from the marking $M_0$ to the marking $M$ via first firing all transitions in $T_{G,\mathit{in}}$ followed by firing all transitions in $T_{G,\mathit{out}}$, but only at very high costs.
	
	While the gadget above is polynomial in the size of the input for bounded systems where we have $d\leq b$, this is not the case for general systems anymore.
	Luckily, a simpler gadget suffices whose size does not depend on the markings $M_0$ and $M$ making it polynomial even if markings are encoded in binary. First, we add a place $p_\aux$ and transitions $t_s$ and $t_e$ with $\pset t_s=\emptyset$, $t_s\pset=\Set{p_\aux}$, $\pset t_e=\Set{p_\aux}$, and $t_e\pset=\emptyset$. Then, for each place $p\in P$, we also add a transition $t^-_p$ with $\pset t^-_p=\Set{p_\aux,p}$ and $t^-_p\pset=\Set{p_\aux}$ as well as a transition $t^+_p$ with $\pset t^+_p=\Set{p_\aux}$ and $t^+_p\pset=\Set{p_\aux,p}$. As soon as $t_s$ fires and places a token in $p_\aux$, this construction allows us to remove tokens in a place $p$ one by one by firing $t^-_p$, or to add tokens by firing $t^+_p$, respectively. Finally, we assign costs of $k+1$ to $t_s$ and costs of $0$ to all other added transitions. Then, $M$ is reachable from $M_0$ by firing $t_s$, suitable transitions $t^-_p$ and $t^+_p$, and finally $t_e$. But, every firing sequence which uses a transition of the gadget has to fire $t_s$ first which in turn incurs costs of at least $k+1$. Note that this gadget does not preserve boundedness and therefore only applicable in the general case.
	
	These adaptions let us now apply the reduction to the classes of bounded and of general systems. Therefore, we can conclude this section with the following two results:
	\begin{corollary}
		\label{cor:align_bounded_pspace_hard}%
		On the class of bounded systems, $\CPalign$ is $\PSPACE$-hard.
	\end{corollary}
	\begin{corollary}
		\label{cor:align_general_ackermann_hard}%
		On the class of general systems, $\CPalign$ is $\Ackermann$-hard.
	\end{corollary}
	
	\section{An Upper Bound for Alignments on Safe Petri Nets}
	\label{sec:safenets}
	
	Now that we have established lower bounds for some classes of Petri nets, we will next derive our first upper bound and show that $\CPalign$ is contained in $\PSPACE$ on the class of safe systems. We begin by showing the reversed direction of our last reduction, namely that $\CPalign$ is reducible to $\CPcostReach$ in polynomial time.
	\begin{lemma}
		\label{lem:align_to_mincostreach}%
		$\CPalign$ is polynomial-time reducible to $\CPcostReach$.
	\end{lemma}
	\begin{proof}
		Let $\Sigma$, $S=(N,M_\init,M_\final)$, $N=(P,T,F,\ell)$, $\ell\colon T\to\Sigmatau$, $\sigma\in\Sigma^*$, $c\colon\mathit{LM_{\mathrm\Sigma,S}}\to\costRange$, and $k$ be an input of the $\CPalign$ problem.
		That is, $\Sigma$ is an alphabet representing the set of activities, $S=(N,M_\init,M_\final)$ is an easy-sound system, i.e., $\scfs{S}\neq\emptyset$, with the Petri net $N=(P,T,F,\ell)$ and labeling function $\ell\colon T\to\Sigmatau$, $\sigma\in\Sigma^*$ is a trace over the alphabet $\Sigma$, and $c\colon\mathit{LM_{\mathrm\Sigma,S}}\to\costRange$ is a function which assigns costs to each legal move between $\Sigma$ and $S$.
		
		According to \autocite{Adriansyah2014}, finding an optimal alignment between $\sigma$ and $S$ is identical to finding a cost-minimal complete firing sequence in the synchronous product net $\mathcal T(\sigma)\otimes S$. Let $\mathcal T(\sigma)\otimes S\coloneq(N',M_\init',M_\final')$ where $N'\coloneq(P',T',F',\ell')$ and in particular $T'\subseteq LM_{\mathrm{\Sigma},S}$. Therefore, a solution to $\CPcostReach$ with a system $(N',M_\init')$ where $N'=(P',T',F',\ell')$, a marking $M_{\final}'$, a cost function $c$, and a number $k$ as input is also a solution to $\CPalign$.
	\end{proof}
	
	Furthermore, we can see that considering costs when firing a transition does not significantly increase the complexity of the reachability problem. In fact, there exists a polynomial-space algorithm that is capable of solving $\CPcostReach$. Using \emph{Savitch's Theorem} \autocite{Savitch1970}, we can now show that $\CPcostReach$ is contained in $\PSPACE$ on the class of safe systems:
	\begin{theorem}
		\label{lem:mincostreach_safe_pspace}%
		On the class of safe systems, $\CPcostReach$ can be decided in polynomial space (in short: $\CPcostReach \in \PSPACE$).
	\end{theorem}
	\begin{proof}
		To find a \emph{deterministic} $\PSPACE$ algorithm for $\CPcostReach$, we use \emph{Savitch's Theorem} \autocite{Savitch1970}: for every \emph{non-deterministic} $\PSPACE$ algorithm, there also exists an equivalent \emph{deterministic} algorithm. Hence it suffices to sketch a non-deterministic algorithm in what follows.
		
		Let $(N,M_0)$, $N=(P,T,F,\ell)$, $c$, $M$, and $t$ be an input of the $\CPcostReach$ problem. That is, $(N,M_0)$ is a safe system with a Petri net $N=(P,T,F,\ell)$, $c\colon T\to\costRange$ is a cost function, $M\in\N^P$ is a marking, and $t\in\costRange$ is a cost limit.
		The algorithm stores a marking $\overline{M}$, which is initially set to $\overline{M}=M_0$, and a cost value $\bar{c}$, initially set to $\bar{c}=0$. Note that since $(N,M_0)$ is safe, a marking of $N$ can be stored in polynomial space. As long as $\overline{M}\neq M$, the algorithm non-deterministically chooses a transition $t\in T$ enabled in marking $\overline{M}$ and computes the marking $\overline{M}'$ after firing $t$, i.e., $(N,\overline{M})\fire t(N,\overline{M}')$. Then, the stored marking is set to $\overline{M}\coloneq\overline{M}'$ and the stored cost value is set to $\bar{c}\coloneq\bar{c}+c(t)$.
		
		If $M$ is not reachable, the algorithm does not necessarily terminate. Thus, we add a counter which counts the number of fired transitions. Because $(N,M_0)$ is safe, it has at most $2^{\abs P}$ reachable markings. Therefore, we can stop if more than $2^{\abs P}-1$ transitions were fired. If the algorithm reaches the marking $M$ and $\bar{c}\leq t$, it stops and $M$ can be reached within the cost limit. If the counter exceeds $2^{\abs P}-1$ or $\bar{c}$ exceeds $t$, $M$ cannot be reached within the cost limit.
	\end{proof}
	
	Since $\CPreach$ on safe systems is $\PSPACE$-complete \autocite{ChengEP1993,ChengEP1995}, we can conclude:
	\begin{lemma}
		\label{thm:mincostreach_safe_pspace_complete}%
		On the class of safe systems, $\CPcostReach$ is $\PSPACE$-complete.
	\end{lemma}
	\begin{proof}
		On the class of safe systems, $\CPreach$ is $\PSPACE$-complete~\autocite{ChengEP1993,ChengEP1995}. By \zcref{lem:reach_to_mincostreach}, $\CPcostReach$ is $\PSPACE$-hard. In combination with \zcref{lem:mincostreach_safe_pspace}, $\CPcostReach$ is $\PSPACE$-complete on the class of safe systems.
	\end{proof}
	
	Altogether, this yields our first main result of this article:
	\begin{theorem}
		\label{thm:align_safe_pspace_complete}%
		On the class of safe systems, $\CPalign$ is $\PSPACE$-complete.
	\end{theorem}
	\begin{proof}
		A trace system is safe by definition, the synchronous product considered in \zcref{lem:align_to_mincostreach} is also safe according to \zcref{cor:synchronous_product_boundedness}. Hence, with \zcref{thm:mincostreach_safe_pspace_complete} we have $\CPalign\in\PSPACE$ on the class of safe systems. Due to \zcref{lem:mincostreach_to_align,thm:mincostreach_safe_pspace_complete}, $\CPalign$ is also $\PSPACE$-hard on the class of safe systems and thus $\PSPACE$-complete.
	\end{proof}
	
	The proof of \zcref{lem:mincostreach_safe_pspace} can be generalized to $b$-bounded systems where $b$ is specified in unary notation. Such a system has at most $(b+1)^{\abs P}$ reachable markings, so that markings and the step counter can still be stored in polynomial space. Together with \zcref{cor:synchronous_product_boundedness,lem:align_to_mincostreach,cor:align_bounded_pspace_hard}, we then obtain:
	\begin{corollary}
		\label{cor:align_bounded_pspace_complete}%
		On the class of bounded systems, $\CPalign$ is $\PSPACE$-complete.
	\end{corollary}
	
	\section{Alignments on Safe and Sound Workflow Nets}
	\label{sec:wfnets}
	
	The syntactic properties of a workflow net alone do not imply any form of \enquote{well-behavedness}, for example, deadlocks are still possible. This is different when soundness is added as an additional requirement.
	In essence, soundness does not only imply boundedness, but also corresponds to the notion of \emph{liveness} studied in Petri net theory. Liveness requires that each transition can eventually fire from every reachable marking. In combination with the structural properties of a workflow net, this rules out deadlocks and guarantees proper completion of processes.
	Of course, liveness cannot hold in a workflow net, simply because we have a terminal sink marking $\Multiset{o}$. However, if we \emph{short-circuit} the workflow net by simply adding a transition $t_\Lreset$, which moves the single token from $o$ back to $i$, i.e., $\pset t_\Lreset = \Set{o}$ and $t_\Lreset\pset = \Set{i}$, then liveness is a valid requirement for the \emph{short-circuited net}. This yields another way to describe soundness of workflow nets:
	\begin{theorem}[{{\normalfont\autocite[Theorem~11]{vanderAalst1997}}}]
		\label{thm:wfnet:sound:live}%
		A workflow net is \emph{sound} if and only if the \emph{short-circuited} net is live and bounded.
	\end{theorem}
	Short-circuiting also allows us to transfer lower bounds for $\CPalign$ from workflow nets to live systems. The key observation is that transitions which are too expensive to be used do not change the answer to an instance of $\CPalign$.
	\begin{lemma}
		\label{lem:shortcircuit_align}%
		Let $N=(P,T,F,\ell,i,o)$ be a workflow net and $\bar{N}$ its short-circuited net with the additional transition $t_\Lreset$ such that $\pset t_\Lreset = \Set{o}$ and $t_\Lreset\pset = \Set{i}$.
		For every trace $\sigma$, cost function $c$, and threshold $k$, there is an alignment between $\sigma$ and $(N,\Multiset{i},\Multiset{o})$ with costs of at most $k$ if and only if there is one between $\sigma$ and $(\bar{N},\Multiset{i},\Multiset{o})$ with costs of at most $k$ with respect to the extension of $c$ that assigns costs of $k+1$ to all moves involving $t_\Lreset$. If $N$ is sound, then $(\bar{N},\Multiset{i})$ is live, bounded, and cyclic, and it has the same reachable markings as $(N,\Multiset{i})$.
	\end{lemma}
	\begin{proof}
		Adding $t_\Lreset$ does not change the pre- and postsets of the transitions in $T$, so every alignment for $N$ is also one for $\bar{N}$ with the same costs. As costs are non-negative, an alignment for $\bar{N}$ with costs of at most $k$ must not contain any move involving $t_\Lreset$ and is therefore also an alignment for $N$. The remainder follows from \zcref{thm:wfnet:sound:live} and \autocite[Lemma~1]{SchwanenP2026}.
	\end{proof}
	
	Here, we consider \emph{safe} and sound workflow nets. Hence, $\CPalign$ is in $\PSPACE$ on this class (cf.\@ \zcref{thm:align_safe_pspace_complete}).
	Despite the further semantic restriction of soundness, the $\PSPACE$ upper bound turns out to be a lower bound as well: we can show that $\CPalign$ is $\PSPACE$-complete on safe and sound workflow nets. To prove $\PSPACE$-hardness, we adapt the well-known hardness construction for reachability on safe Petri nets.
	Specifically, we encode the computation of a $\PSPACE$-Turing machine using a safe Petri net: a polynomial number of places is used to store the current configuration of the machine (via markings), while for each possible computational step of the Turing machine (depending on the current state, the head position, and the read symbol), a distinct transition can be fired which produces the unique successor configuration.
	While the construction itself yields a safe Petri net, soundness is not guaranteed. Indeed, the resulting Petri net can contain several dead transitions simply because certain configurations cannot be reached, i.e., those transitions can never be fired from the initial marking (which contradicts liveness).
	We thus extend the construction by new gadgets to overcome this difficulty.
	
	\begin{theorem}[$\PSPACE$-Hardness of Alignments on Safe and Sound Workflow Nets]
		\label{lem:wfnets:pspace}%
		There is a polynomial time algorithm which transforms a deterministic Turing machine $\M$ with polynomial space bound $p(n)$ and an input $w$ into a safe and sound workflow net $S=(N, M_\init, M_\final)$ and a trace $\sigma$ such that, with respect to the standard cost function, $S$ and $\sigma$ can be aligned with $0$ costs if and only if $\M$ accepts $w$.
	\end{theorem}
	\begin{proof}
		We extend the construction in~\autocite[Theorem~4]{ChengEP1993}.
		First, we make some assumptions on $\M$ which can be guaranteed by preprocessing. When the (deterministic, single tape) Turing machine $\M$ is started with input $w$, during the computation, the head of $\M$ only moves between positions $0$ and $p(n)$, where $n = \abs w$, starting at position $0$, i.e., it will never move beyond the bounds of cell $0$ and cell $p(n)$.
		Moreover, each computation is finite, i.e., the machine $\M$ never repeats a configuration. In particular, $\M$ halts on every input and either accepts or rejects.
		Finally, there is precisely one accepting and one rejecting configuration of the machine $\M$. To guarantee this, one can implement a subroutine such that, whenever the machine $\M$ enters a final state (accepting or rejecting), then the tape is cleared, the head moves back to position $0$, and the machine enters a unique accepting or rejecting state, respectively.
		
		Let $\M = (K, \Sigma, \Gamma, \delta, q_0, q_+, q_-, \bot)$ be a deterministic Turing machine where $K$ is the set of states, $\Sigma$ the input alphabet, $\Gamma$ the tape alphabet, $\delta\colon K\setminus\Set{q_+, q_-} \times \Gamma \to K \times \Gamma \times \Set{-1, 1, 0}$ the transition function ($-1$ means the head moves one position to the left, $1$ means the head moves one position to the right, and $0$ means the head stays in position), $q_0\in K$ the initial state, $q_+ \in K$ the unique accepting state, $q_- \in K$ the unique rejecting state, and $\bot\in\Gamma\setminus\Sigma$ the blank symbol.
		
		We describe the construction in two steps. First, we simulate $\M$ via a workflow net. Second, we rule out dead transitions (i.e., transitions which are never enabled) to guarantee soundness.
		To encode the computation of $\M$ on input $w$, we define a workflow net $N = (P, T, F, \ell, p_\init, p_\final)$ with a set of places $P$ consisting of:
		\begin{itemize}
			\item the initial place $p_\init$ and the final place $p_\final$,
			\item a place $p^K_q$ for each state $q \in K$ (a token in $p^K_q$ indicates that in the current configuration $\M$ is in state $q$),
			\item a place $p^H_i$ for each possible head position $i \in \Set{0, \dots, p(n)}$ (a token in $p^H_i$ indicates that in the current configuration the head is at position $i$),
			\item a place $p^C_{i, a}$ for each possible tape cell content, i.e., for each combination of a valid position $i \in \Set{0, \dots, p(n)}$ and a tape symbol $a \in \Gamma$ (a token in $p^C_{i, a}$ indicates that in the current configuration tape cell $i$ holds symbol $a$).
		\end{itemize}
		With this preparation, we identify configurations of $\M$ with markings of $N$.
		To simulate the computation, we introduce the following set of transitions $T$:
		\begin{itemize}
			\item One transition $t_\Lstart$ that yields the initial configuration and triggers the simulation, i.e., $\pset t_\Lstart = \Set{p_\init}$ and
			$t_\Lstart\pset = \Set{p^K_{q_0}} \cup \Set{p^H_0} \cup \Set{p^C_{i,w_i}\given i < \abs w} \cup \Set{p^C_{i, \bot}\given \abs w \leq i \leq p(n)}$.
			\item One transition $t_\Lacc$ that finalizes the computation when the (unique) accepting configuration is reached, i.e., $t_\Lacc\pset = \Set{p_\final}$ and $\pset t_\Lacc = \Set{p^K_{q_+}} \cup \Set{p^H_0} \cup \Set{p^C_{i, \bot}\given i \leq p(n)}$. This transition is labeled by $\Lacc$ and there is no other transition labeled by $\Lacc$.
			In the same way, we add one transition $t_\Lrej$ that finalizes the computation when the (unique) rejecting configuration is reached.
			\item For each possible computational step, we introduce a distinct transition: for each state $q \in K$, symbol $a \in \Gamma$ with $\delta(q,a) = (q',b,d)$, and for each head position $i \in \Set{0, \dots, p(n)}$, we introduce a transition $t[q,a,i]$ that models the configuration change which occurs when $\M$ is in state $q$, the head is at position $i$, and reads the symbol $a$. More precisely, $\pset t[q,a,i] = \Set{p^K_q, p^H_i, p^C_{i,a}}$ and $t[q,a,i]\pset = \Set{p^K_{q'}, p^H_{i + d}, p^C_{i,b}}$.
		\end{itemize}
		
		By construction, we can trigger the simulation from the initial marking $M_\init=\Multiset{p_\init}$ by firing $t_\Lstart$. This generates the initial configuration of the computation of $\M$ on $w$. Since the machine $\M$ is deterministic, from that point onward there is at most one transition of the form $t[q,a,i]$ that can be fired in the current marking/configuration. This transition, in turn, updates the marking/configuration according to the transition function of $\M$. Since we have prepared $\M$ in such a way that the computation is acyclic, we will finally reach the unique accepting or rejecting configuration from which we can fire $t_\Lacc$ or $t_\Lrej$, respectively, to reach the final marking $M_\final=\Multiset{p_\final}$.
		Giving the one-to-one correspondence between markings of $N$ and configurations of $\M$, the resulting net is safe as every reachable marking corresponds to a configuration in the sense described above and such markings only hold at most one token per place.
		
		The problem is that $N$ is not sound. In general, it might be that a transition of the form $t[q,a,i]$ can never fire simply because the computation of $\M$ on $w$ does never run into an enabling configuration.
		Also, we can either fire $t_\Lacc$ or $t_\Lrej$ from the initial marking, but not both.
		
		To overcome this technical problem, we extend the workflow net further to explicitly allow the activation of each transition. Of course, we have to take care that we do not lose the safeness property on the way.
		First, we add a new place $p_\aux$ indicating when the net is in \emph{auxiliary} mode.
		Then, we add two new transitions $t^0_\Lacc$ and $t^0_\Lrej$ which can fire at the initial marking and activate $t_\Lacc$ and $t_\Lrej$: we set $\pset t^0_\Lacc = \Set{p_\init}$ and $ t^0_\Lacc\pset = \pset t_\Lacc$ and analogously for $t^0_\Lrej$.
		In other words, $t^0_\Lacc$ and $t^0_\Lrej$ produce the two unique markings where precisely $t_\Lacc$ or $t_\Lrej$ is enabled, respectively. So, we can reach the final marking by firing either of them. Note that transitions of the form $t[q,a,i]$ are not enabled in these terminal configurations of $\M$.
		
		Second, for each transition $t[q,a,i]$ we add two new transitions $t^0[q,a,i]$ and $t^1[q,a,i]$ which activate and deactivate $t[q,a,i]$ from the initial marking, i.e.,
		\begin{alignat*}{9}
			\pset &t^0[q,a,i] &{}={}&\mathrlap{\Set{p_\init},} &\quad\text{  and  }\quad& &t^0[q,a,i] \pset &{}={}& &\Set{p_\aux} \cup{} \pset t[q,a,i], \\
			\pset &t^1[q,a,i] &{}={}&\Set{p_\aux} \cup{} t[q,a,i] \pset, &\quad\text{  and  }\quad& &t^1[q,a,i] \pset &{}={}& &\mathrlap{\Set{p_\final}.}
		\end{alignat*}
		In this way, we can move via the sequence $\seq{t^0[q,a,i], t[q,a,i], t^1[q,a,i]}$ from the initial to the final marking and fire $t[q,a,i]$ along the way.
		However, there is one subtlety we have to discuss.
		In contrast to the case of $t_\Lacc$ and $t_\Lrej$, the manual activation and firing of $t[q,a,i]$ might enable transitions different from $t^1[q,a,i]$. In fact, if the head does not move in state $q$ while reading $a$, three tokens would be produced by $t[q,a,i]$ which would allow the net to fire another transition of the form $t[q',b,i]$. But, all such transitions are conservative in the sense that they do not alter the total count of tokens (they all consume and produce three tokens). Thus, the total count of tokens remains three which means that we will never be able to fire one of the final transitions $t_\Lacc$ or $t_\Lrej$. Since $\M$ does not allow cyclic computations, eventually we get stuck in the simulation component after firing one transition of the form $t[q',b,i]$ which moves the head to the left or right because for the new cell, we are missing a token in a place $p^C_{a,i+d}$ for $d\in\Set{-1,1}$.
		In this setting we can simply fire $t^1[q',b,i]$ which removes the tokens produced by $t[q',b,i]$ and enters the final marking.
		Finally, all introduced auxiliary transitions get the extra label $\Laux$.
		Note that we cannot mix a proper simulation of the computation of $\M$ on $w$ triggered by firing $t_\Lstart$ with a transition of the form $t^1[q,a,i]$ since this transition requires a token in $p_\aux$ in order to become enabled.
		Altogether, by this second extension each transition in the workflow net can be fired from the initial marking and we maintain the property to always reach the final marking in the end.
		
		To complete our reduction to the alignment problem, let $\sigma\coloneq\seq\Lacc$. Then, we claim that $\sigma$ and the safe and sound workflow net $N$ can be aligned with costs $0$ (with respect to the standard cost function) if and only if $\M$ accepts input $w$.
		Clearly, if $\M$ accepts $w$, the simulation triggered by firing $t_\Lstart$ simulates the computation via silent transitions of the form $t[q,a,i]$ until the single transition $t_\Lacc$ with label $\Lacc$ can eventually be fired synchronously with the symbol $\Lacc$ in $\sigma$. This does not incur any costs since we only have one synchronous move and several silent moves in the workflow net.
		If, on the other hand, the computation of $\M$ on $w$ is rejecting, there is no way to align $\sigma$ with costs $0$. In fact, if any of the auxiliary transitions is used, this will immediately lead to a model move since $\sigma$ does not contain the symbol $\Laux$. If, on the other hand the proper simulation of $\M$ on $w$ is started via $t_\Lstart$, then this will end up in a completely silent run ending with $t_\Lrej$ and we would require a log move for the symbol $\Lacc$ in the trace.
		
		Finally, from the description it is immediate that the construction can be carried out by a polynomial-time algorithm.
	\end{proof}
	In \autocite{SchwanenP2026}, we also sketched how to adapt the proof above to show that $\CPreach$ on safe and sound workflow nets is already $\PSPACE$-hard.
	
	Together with \zcref{thm:align_safe_pspace_complete} we can conclude:
	\begin{theorem}
		\label{thm:align_wfnets_pspace_complete}%
		On the class of safe and sound workflow nets, $\CPalign$ is $\PSPACE$-complete.
	\end{theorem}
	As an immediate consequence, we also get:
	\begin{corollary}
		On the class of safe and sound accepting systems, $\CPalign$ is $\PSPACE$-complete.
	\end{corollary}
	
	Finally, since short-circuiting a safe and sound workflow net yields a cyclic, live, and safe system (\zcref{lem:shortcircuit_align}), \zcref{thm:align_wfnets_pspace_complete,thm:align_safe_pspace_complete} also yield:
	\begin{corollary}
		\label{cor:align_cyclic_pspace_complete}%
		On the class of cyclic, live, and safe systems, $\CPalign$ is $\PSPACE$-complete.
	\end{corollary}
	
	\section{An Upper Bound for Alignments on Live, Bounded, Free-Choice Systems}
	\label{sec:lbfc_systems}
	
	As we saw, to improve algorithmic bounds for $\CPalign$, imposing soundness alone is insufficient. Thus, we turn our attention to \emph{live, $b$-bounded, free-choice systems} (\lbfc-systems, for short).
	These nets have been thoroughly studied in the area of Petri net theory and enjoy nice structural properties~\autocite[cf.][]{DeselE1995}.
	
	In our main result (\zcref{thm:lbf:poly:alignments}), we show that for each trace and each \lbfc-system there exists an optimal alignment of polynomial length. From this, we obtain a simple guess-and-verify $\NP$-algorithm for the alignment problem which reduces the complexity from $\PSPACE$ to $\NP$.
	
	The result is interesting in its own right.
	For example, if the costs for moves are bounded polynomially in the size of the system (which is expected for \enquote{reasonable} cost functions, e.g., for the standard cost function), then the costs of optimal alignments are also bounded polynomially. This is clearly helpful for computing quality metrics and it could also be exploited within search algorithms for earlier pruning of the exponential state space or within mixed integer linear program encodings.
	
	\begin{definition}[\lbfc-System]
		We denote a system $(N,M_0)$ which is live, ($b$-)bounded, and free-choice as \emph{\lbfc-system}.
	\end{definition}
	In this article, we consider \lbfc-systems with respect to a \emph{fixed} value of $b\in\N$, i.e., the place bound $b$ may vary across different \lbfc-systems, but is of course fixed for each individual \lbfc-system. Implicitly, we think of $b = 1$ (safe systems), but any other fixed value for $b$ is possible, too.
	
	Our $\NP$-algorithm for \lbfc-systems is based on the following key property: whenever a marking $M'$ can be reached from $M$ in an \lbfc-system, then there \emph{exists some} connecting firing sequence of polynomial length. This was first shown in~\autocite{DeselE1995a}, but we rely on the textbook by the same authors~\autocite{DeselE1995}.
	
	\begin{theorem}[Shortest Sequence Theorem {\normalfont\autocite[Theorem~9.17]{DeselE1995}}]
		Let $(N,M_0)$ be an \lbfc-system with $n$ transitions and let $M$ be a reachable marking.
		Then, there is a firing sequence $\sigma$ such that $(N,M_0)\fire\sigma(N,M)$ and
		\begin{equation*}
			\abs\sigma \leq b\cdot\frac{n\cdot(n+1)\cdot(n+2)}{6}.
		\end{equation*}
	\end{theorem}
	It is important to note that this result does \emph{not} hold for general $b$-bounded Petri net systems: here, shortest connecting paths can be of exponential length.
	This is even true for safe systems which can simulate $\PSPACE$-computations as we saw earlier.
	Unfortunately, when we tried to apply the Shortest Sequence Theorem to alignments, we faced a technical obstacle: the theorem only guarantees the \emph{existence} of a \enquote{short} firing sequence.
	In particular, when we start from an arbitrary firing sequence and then, using the above theorem, pass over to a short one, this new sequence can be completely different from the original firing sequence. This is problematic in context of alignments, since the \enquote{short} firing sequence could contain different transitions, potentially incurring much higher costs of moves.
	In a nutshell, a \enquote{shorter} sequence does not necessarily qualify as a \enquote{cheaper} sequence.
	
	Luckily, by carefully analyzing the proof in~\autocite{DeselE1995}, we found that the result can be stated in a more general form. For two firing sequences $\sigma_1, \sigma_2 \in T^*$ we write $\vec{\sigma_1}\leq\vec{\sigma_2}$ if the \emph{multi}set of transitions occurring in $\sigma_1$ is a (\emph{multi}-)subset of the \emph{multi}set of transitions occurring in $\sigma_2$. In other words, $\sigma_2$ can be transformed into $\sigma_1$ by permuting and deleting transitions.
	
	\begin{theorem}[Shortest Sequence Theorem---Generalized Form]
		\label{thm:shortestst:tsystems:rev}%
		Let $(N,M_0)$ with $N=(P,T,F,\ell)$ be an \lbfc-system with bound $b$.
		Moreover, let $\sigma\in T^*$ and $M_1,M_2\in\reach{N,M_0}$ such that $(N,M_1)\fire\sigma(N,M_2)$.
		Then, there exists a firing sequence $\sigma'\in T^*$ with $\vec{\sigma'}\leq\vec{\sigma}$ such that $(N,M_1)\fire{\sigma'}(N,M_2)$ and
		\begin{equation*}
			\abs{\sigma'} \leq b \cdot \frac{\abs T \cdot (\abs T+1) \cdot (\abs T+2)}{6}.
		\end{equation*}
	\end{theorem}
	
	In order to prove \zcref{thm:shortestst:tsystems:rev} we go through the machinery of~\autocite{DeselE1995} while making necessary changes to results and proofs. For those parts that do not require any adaptation, we refer to~\autocite{DeselE1995} for details.
	
	For a firing sequence $\sigma \in T^*$, we denote by $\supp{\vec{\sigma}}\subseteq T$ the set of transitions that occur in $\sigma$.
	The proof consists of two steps.
	First, we consider the case of \emph{T-systems}. T-systems are (a very simple form of) free-choice systems where each place has at most one successor transition. In such systems, whenever a transition is enabled, it cannot be disabled by firing any other transition.
	\begin{definition}[T-Net, T-System]
		\label{def:tsystem}%
		A Petri net $N=(P,T,F,\ell)$ is a \emph{T-net} if each place has at most one input and one output transition, i.e., $\forall p\in P\colon \abs{\pset p}\leq1, \abs{p\pset}\leq1$.
		A system $(N,M_0)$ is a \emph{T-system} if $N$ is a T-net.
	\end{definition}
	The syntactic restrictions of T-nets allow us to rearrange firing sequences in such a way that the same sets of transitions are repeatedly fired until, eventually, more and more transitions die out and the set of occurring transitions becomes smaller and smaller (assuming $b$-boundedness).
	Let us start with some important definitions (all taken from~\autocite{DeselE1995}).
	
	\begin{definition}[Biased Firing Sequence {\normalfont\autocite[cf.][Definition~3.22]{DeselE1995}}]
		A firing sequence $\sigma\in T^*$ is \emph{biased} if for all $t_1,t_2\in\supp{\vec{\sigma}}$, $t_1\neq t_2$, it holds that $\pset t_1\cap\pset t_2=\emptyset$.
	\end{definition}
	
	Note that firing sequences in T-systems are always biased. In biased sequences, we can move each occurring transition to an initial segment:
	\begin{lemma}[{{\normalfont\autocite[Lemma~3.24]{DeselE1995}}}]
		\label{lem:biased:permute:once}%
		Let $(N,M_0)$ be a system with $N=(P,T,F,\ell)$ and let $\sigma\in T^*$ be a biased firing sequence. Then, there exists a permutation $\rho=\sigma_1\sigma_2$ of $\sigma$ with $\vec{\rho}=\vec{\sigma}$ such that $\exists M\in\N^P\colon(N,M_0)\fire{\sigma}(N,M)\wedge(N,M_0)\fire{\sigma_1\sigma_2}(N,M)$ and no transition occurs more than once in $\sigma_1$ and every transition that occurs in $\sigma_2$ occurs also in $\sigma_1$, i.e., $\supp{\vec{\sigma_2}} \subseteq \supp{\vec{\sigma_1}}$.
	\end{lemma}
	
	What happens if we apply this lemma repeatedly? Then, we obtain a permutation $\sigma_1\sigma_2\sigma_3\cdots$ of the original firing sequence such that each $\sigma_i$ contains each transition at most once and the sets of occurring transitions $\supp{\vec{\sigma_i}}$ in the subsequences $\sigma_i$ get smaller and smaller with increasing index $i$, i.e., $\supp{\vec{\sigma_1}} \supseteq \supp{\vec{\sigma_2}} \supseteq \supp{\vec{\sigma_3}} \supseteq \cdots$.
	In order to guarantee the existence of \enquote{short} sequences, the idea is to bound the number of repetitions of the \emph{same} set of occurring transitions. In fact, the next lemma shows that after at least $b$ many repetitions the set of occurring sequences has to decrease (if not, we can shorten the firing sequence).
	For the following lemma we need to generalize the result from~\autocite{DeselE1995}:
	
	\begin{lemma}[{{\normalfont Generalized Form of~{\autocite[Lemma~3.25]{DeselE1995}}}}]
		\label{lem:bounded:bias:decrease}%
		Let $(N,M_0)$ with $N=(P,T,F,\ell)$ be a $b$-bounded system and let $(N,M_0)\fire\sigma(N,M)$ for a firing sequence $\sigma\in T^*$ and a marking $M$ such that $\sigma$ is biased and non-empty.
		Then, there exists a firing sequence $\rho=\sigma_1\sigma_2\in T^*$ with $\vec{\rho}\leq\vec{\sigma}$ such that
		\begin{itemize}
			\item $(N,M_0)\fire\sigma(N,M)$ and $(N,M_0)\fire{\sigma_1\sigma_2}(N,M)$,
			\item each transition occurs at most $b$ times in $\sigma_1$, and
			\item $\supp{\vec{\sigma_1}} \supset \supp{\vec{\sigma_2}}$ (set of occurring transitions decreases).
		\end{itemize}
	\end{lemma}
	\begin{proof}
		First, we repeatedly apply \zcref{lem:biased:permute:once} in order to obtain a permutation of $\sigma$ which is of the form $\rho_1 \rho_2 \cdots \rho_n$ with $\rho_i\neq\seq{}$ with the following properties:
		\begin{itemize}
			\item for all $i=1,\dots,n$, no transition occurs more than once in $\rho_i$, and
			\item for all $i < n$, $\supp{\vec{\rho_i}} \supseteq \supp{\vec{\rho_{i+1}}}$.
		\end{itemize}
		If $n \leq b$, we are done: simply choose $\sigma_2 = \seq{}$ and $\sigma_1 = \rho_1 \cdots \rho_n$.
		So, let us assume $n \geq b+1$.
		If $\supp{\vec{\rho_1}} \supset \supp{\vec{\rho_{b+1}}}$, we can stop our argument at this point as well, since then $\sigma_1 = \rho_1 \cdots \rho_b$ and $\sigma_2 = \rho_{b+1} \cdots \rho_n$ would have the desired properties.
		So, let us also assume that $\supp{\vec{\rho_1}} = \supp{\vec{\rho_{b+1}}}$.
		Choose the maximal $m$ such that $\supp{\vec{\rho_1}} = \supp{\vec{\rho_m}}$ (note that $b + 1 \leq m \leq n$).
		Let $X \coloneq \supp{\vec{\rho_1}} = \supp{\vec{\rho_m}} \subseteq T$.
		Since each transition in $X$ occurs precisely once in each of the $\rho_i$, we know that for each place the same number of tokens is added or removed after firing each of the subsequences $\rho_i$. Since each place can hold at most $b$ tokens (since the net is $b$-bounded) this means that firing all transitions in $X$ must leave the number of tokens in each place invariant.
		But this, in turn, implies that each of the $\rho_i$, $1 \leq i \leq m$ does not modify the initial marking $M_0$.
		Hence, we can simply choose $\sigma_1 = \rho_1$ and $\sigma_2 = \rho_{m+1} \cdots \rho_n$ which completes our proof (note that this is the case where we can shorten the initial sequence).
	\end{proof}
	
	This already implies the Shortest Sequence Theorem for T-systems. For later use, let us state the concrete implications in a generalized form of the original \emph{Biased Sequence Lemma}~\autocite[Lemma~3.26]{DeselE1995}:
	\begin{lemma}[Biased Sequence Lemma, {\normalfont Generalized Form of~{\autocite[Lemma~3.26]{DeselE1995}}}]
		\label{lem:lbf:biased:sequence}%
		Let $(N,M_0)$ with $N=(P,T,F,\ell)$ be a $b$-bounded system and let $(N,M_0)\fire\sigma(N,M)$ for a firing sequence $\sigma\in T^*$ and a marking $M$ such that $\sigma$ is biased and non-empty.
		Let $k$ denote the number of (distinct) transitions in $\sigma$, i.e., $k\coloneq\abs{\supp{\vec{\sigma}}}$.
		Then, there exists a firing sequence $\rho\in T^*$ such that $\vec{\rho}\leq\vec{\sigma}$, $(N,M_0)\fire\rho(N,M)$, and
		\begin{equation*}
			\abs\rho \leq b \cdot \frac{k \cdot (k+1)}{2}.
		\end{equation*}
	\end{lemma}
	\begin{proof}
		By repeatedly applying \zcref{lem:bounded:bias:decrease}, we find $\rho = \rho_1 \rho_2 \cdots \rho_{k}$, $\vec{\rho} \leq \vec{\sigma}$ where
		\begin{itemize}
			\item $(N,M_0)\fire{\rho_1 \rho_2 \cdots \rho_k}(N,M)$,
			\item each transition occurs at most $b$ times in $\rho_i$, and
			\item $\rho_i$ contains at most $(k+1-i)$ many different transitions.
		\end{itemize}
		Hence, $\abs\rho \leq b \cdot \sum_{i=1}^k i = b \cdot k \cdot (k+1)/2$ as claimed.
	\end{proof}
	
	Next, we consider the case of general free-choice systems.
	We start by describing some intuition of the overall argument. In free-choice systems, it is no longer the case that firing sequences are biased, so we cannot simply apply the Biased Sequence Lemma.
	Indeed, in free-choice systems, choices are possible and transitions can disable each other, however, only in a very pure form: if two transitions share a common place as precondition, then they actually share \emph{all} their places as preconditions. Hence, choices can occur, but only between transitions with the same preconditions.
	The main idea is to linearize such choices.
	For example, let us say we have three transitions $t_1,t_2,t_3$ with the same set of preconditions which occur in a firing sequence $\sigma$. Then the idea is to rearrange $\sigma$ in such a way that, for instance, all occurrences of $t_1$ happen before the occurrences of $t_2$ and those occurrences happen before the occurrences of $t_3$. If we can achieve this kind of rearrangement, then we can split the whole (rearranged) sequence up into biased subsequences according to the linearization. This, in turn, allows us to apply the Biased Sequence Lemma to the subsequences individually and to prove the Shortest Sequence Theorem for \lbfc-systems.
	
	To proceed, we need to introduce a couple of notions and definitions. In what follows, we implicitly refer to some \lbfc-system $(N,M_0)$ with $N=(P,T,F,\ell)$.
	
	\begin{definition}[Cluster {\normalfont\autocite[cf.][Definition~4.4]{DeselE1995}}]
		For two transitions $t,t'\in T$ we let $t\sim t'$ if $\pset t=\pset t'$.
		This gives rise to an equivalence relation $\sim$ on $T$ whose equivalence classes $[t]$ are called \emph{clusters}.
	\end{definition}
	Strictly speaking, clusters---as per their usual definition---also include the input places of the transitions contained, but this is irrelevant here and therefore disregarded.
	
	\begin{definition}[Conflict Order {\normalfont\autocite[cf.][Definition~9.8]{DeselE1995}}]
		A \emph{conflict order} ${\preceq}\subseteq T\times T$ is any partial order on $T$ such that two transitions $t,t'\in T$ are comparable if and only if $[t]=[t']$, i.e., $\pset t\cap\pset t'\neq\emptyset$.
		The corresponding strict partial order is denoted by $\prec$, i.e., $t\prec t'$ if and only if $t\preceq t'$ and $t\neq t'$.
	\end{definition}
	
	To put it differently, a conflict order is composed of separate linear orderings on the clusters of the \lbfc-system.
	The key idea (and main challenge) in the proof of the \emph{Shortest Sequence Theorem} is to show that each firing sequence can be rearranged in such a way that the permuted firing sequence \emph{is ordered} with respect to \emph{some} conflict order with which the firing sequence \emph{agrees}:
	\begin{definition}[Ordered Firing Sequence {\normalfont\autocite[cf.][Definition~9.8]{DeselE1995}}]
		A firing sequence $\sigma\in T^*$ \emph{is ordered} with respect to a conflict order $\preceq$ if for all $t\prec t'$ there is no occurrence of $t$ in $\sigma$ after an occurrence of $t'$.
		Furthermore, $\sigma$ \emph{agrees} with $\preceq$ if for each cluster $c$ either no transition of $c$ occurs in $\sigma$ or the last transition of $c$ that occurs in $\sigma$ is the maximal transition in $c$ (according to $\preceq$).
	\end{definition}
	
	The central result is the following proposition from~\autocite{DeselE1995} which, conveniently, we can use without modification:
	\begin{proposition}[{{\normalfont\autocite[Proposition~9.16]{DeselE1995}}}]
		\label{thm:lbf:ordered:sequences}%
		Let $(N,M_0)$ with $N = (P,T,F,\ell)$ be an \lbfc-system.
		Moreover, let $(N,M_0)\fire\sigma(N,M)$ for a firing sequence $\sigma\in T^*$ and let $\preceq$ be any conflict order which agrees with $\sigma$.
		Then, there exists a $\preceq$-ordered permutation $\rho\in T^*$ of $\sigma$, i.e., $\vec{\rho}=\vec{\sigma}$, such that $(N,M_0)\fire\rho(N,M)$.
	\end{proposition}
	
	With \zcref{thm:lbf:ordered:sequences} and the Biased Sequence Lemma (\zcref{lem:lbf:biased:sequence}) we can finally prove our generalized version of the Shortest Sequence Theorem.
	\begin{proofof}{\zcref{thm:shortestst:tsystems:rev}}
		First, note that for $M\in\reach{N,M_0}$ the system $(N,M)$ is an \lbfc-system as well. Hence, it suffices to shorten a firing sequence $\sigma$ of the form $(N,M_0)\fire\sigma(N,M)$.
		Moreover, due to \zcref{thm:lbf:ordered:sequences}, we can assume that $\sigma$ is $\preceq$-ordered for some conflict order $\preceq$ (if not, we can permute $\sigma$ accordingly).
		
		We iteratively split $\sigma$ up into parts $\sigma_i$, $i=1,\dots,k$, i.e., $\sigma=\sigma_1\sigma_2\cdots\sigma_k$, with respect to the following property: $\sigma_i$ is a maximal prefix of $\sigma_i\cdots\sigma_{k}$ such that $\sigma_i$ is biased.
		We claim that the number of distinct transitions in $\sigma_{i+1}$ is smaller than the number of distinct transitions in $\sigma_i$.
		This is because the first transition $t'$ of $\sigma_{i+1}$ must be in the same cluster $t'\in[t]$ of some transition $t\neq t'$ that occurs in $\sigma_i$ (because of the maximality of $\sigma_i$). Since $\sigma$ is ordered, $t\prec t'$ and $t$ cannot occur in $\sigma_{i+1}\cdots\sigma_k$.
		Hence, $\sigma_i$ contains at most $(\abs T-i+1)$ distinct transitions.
		In particular, $k\leq\abs T$.
		We have $(N,M_{i-1})\fire{\sigma_i}(N,M_i)$ where $M_k=M$.
		By \zcref{lem:lbf:biased:sequence} we find firing sequences $\sigma_i'$, $\vec{\sigma_i'}\leq\vec{\sigma_i}$ of length at most $b\cdot(\abs T-i+1)\cdot(\abs T-i+2)/2$ with $(N,M_{i-1})\fire{\sigma_i'}(N,M_i)$.
		The shortened sequence $\sigma'=\sigma_1'\sigma_2'\cdots\sigma_k'$, $\vec{\sigma'}\leq\vec{\sigma}$ has length at most
		\begin{equation*}
			\frac{b}{2}\cdot\sum_{i=1}^{\abs T}(\abs T-i+1)\cdot(\abs T-i+2)=\frac{b}{2}\cdot\frac{\abs T\cdot(\abs T+1)\cdot(\abs T+2)}{3},
		\end{equation*}
		and satisfies $(N,M_0)\fire{\sigma'}(N,M)$, which completes our proof.
	\end{proofof}
	
	With this result, we can bound the length of sequences of consecutive model moves in alignments. This, in turn, allows us to show that \lbfc-systems always have optimal alignments of polynomial length:
	\begin{theorem}[Alignments in \lbfc-Systems]
		\label{thm:lbf:poly:alignments}%
		Let $S=(N,M_\init,M_\final)$ with $N=(P,T,F,\ell)$ and labeling function $\ell\colon T\to\Sigmatau$ be an accepting \lbfc-system with bound $b$ and let $\sigma\in\Sigma^*$ be a trace over the alphabet $\Sigma$.
		Then, there exists an optimal alignment $\gamma\in\Gamma_{\sigma,S}$ between $\sigma$ and $S$ such that
		\begin{equation*}
			\abs\gamma \leq (\abs\sigma + 1) \cdot \left(b \cdot \frac{\abs T \cdot (\abs T+1) \cdot (\abs T+2)}{6} + 1\right).
		\end{equation*}
	\end{theorem}
	\begin{proof}
		Let $\gamma=\seq{\gamma_i}_{i=1}^{\abs\gamma}\in\Gamma_{\sigma,S}$ be an optimal alignment of minimal length. We show that $\abs\gamma$ satisfies the above inequality.
		
		Let us denote the length of the trace $\sigma$ by $q\coloneq\abs\sigma$.
		Since $\pi_1(\gamma)\vert_\Sigma=\sigma$, we can find a subsequence $\seq{\gamma_i}_{i\in I}$ of $\gamma$ with $I\subseteq\Set{1,\dots,\abs\gamma}$, $I=\Set{i_1,i_2,\dots,i_q}$, $i_1<i_2<\dots<i_{q}$ and such that $\pi_1(\seq{\gamma_i}_{i\in I})=\sigma$. We use the positions $i_1,i_2,\dots,i_{q}\in I$ in order to split $\gamma$ up into $q+1$ parts:
		\begin{align*}
			\delta_0 &= \seq{\gamma_1,\dots,\gamma_{i_1}}, \\
			\delta_1 &= \seq{\gamma_{i_1+1},\dots,\gamma_{i_2}}, \\
			&\enspace\vdots \\
			\delta_{q-1} &= \seq{\gamma_{i_{q-1}+1},\dots,\gamma_{i_{q}}}, \\
			\delta_{q} &= \seq{\gamma_{i_{q}+1},\dots,\gamma_{\abs\gamma}}.
		\end{align*}
		
		For an illustration, see \zcref{fig:splitting:gamma}.
		Next, we show that for each $j\in\Set{0,\dots,q}$ we have $\abs{\delta_j}\leq b\cdot\abs T\cdot(\abs T+1)\cdot(\abs T+2)/6+1$. If we can verify this, the claim follows.
		Pick some $j\in\Set{0,\dots,q}$ and let
		\begin{align*}
			M_j\in\N^P\colon\ & (N,M_\init)\fire{\pi_2(\seq{\gamma_1,\dots,\gamma_{i_j}})\vert_T}(N,M_j), \\
			M_j'\in\N^P\colon\ & (N,M_\init)\fire{\pi_2(\seq{\gamma_1,\dots,\gamma_{i_{j+1}-1}})\vert_T}(N,M_j'),
		\end{align*}
		where for the cases $j=0$ we let $M_0\coloneq M_\init$ and for $j=q$ we let $M_{q}'\coloneq M_\final$.
		In words, $M_j$ is the marking that we obtain from $M_\init$ by firing the transitions in $\gamma$ of the first $j$ parts, i.e., $\pi_2(\seq{\gamma_1,\dots,\gamma_{i_j}})\vert_T$, and $M_j'$ is the marking that we get by firing the transitions in the first $j+1$ parts, i.e., $\pi_2(\seq{\gamma_1,\dots,\gamma_{i_{j+1}-1}})\vert_T$, except for the very last transition from the move $\gamma_{i_j}$.
		The motivation for looking at these two markings is as follows:
		\begin{itemize}
			\item by definition, we can reach the marking $M_j'$ from marking $M_j$ by firing the intermediate sequence $\pi_2(\seq{\gamma_{i_j+1},\dots,\gamma_{i_{j+1}-1}})$,
			\item each of the moves in this sequence $\seq{\gamma_{i_j+1},\dots,\gamma_{i_{j+1}-1}}$ is of the form $(\nomove,t)$, i.e., we only move in the system, but not in the trace,
			\item the length of this sequence is $\abs{\delta_j}-1$.
		\end{itemize}
		Hence, it suffices to show that the sequence $\seq{\gamma_{i_j+1},\dots,\gamma_{i_{j+1}-1}}$ is of length at most $b\cdot\abs T\cdot(\abs T+1)\cdot(\abs T+2)/6$.
		To see this, we make use of \zcref{thm:shortestst:tsystems:rev}.
		In fact, by this result we know that from $\pi_2(\seq{\gamma_{i_j+1},\dots,\gamma_{i_{j+1}-1}})$ we could construct, by deleting and rearranging transitions, a firing sequence which leads from $M_j$ to $M_j'$ in the underlying system of length at most $b\cdot\abs T\cdot(\abs T+1)\cdot(\abs T+2)/6$.
		Since \emph{all} moves in $\seq{\gamma_{i_j+1},\dots,\gamma_{i_{j+1}-1}}$ are of the form $(\nomove,t)$, we could lift the necessary deletion and rearrangement steps to the level of $\gamma$ without doing any harm to the alignment properties.
		Also note that since we only delete and rearrange moves, the costs of the alignment do not increase.
		Since $\gamma$ was chosen to be an optimal alignment of minimal length, the claim follows.
		\begin{figure}[t]
			\centering
			\newdimen\boxlength
			\boxlength=2em
			\begin{tikzpicture}[scale=.80, transform shape,x=\boxlength,y=\boxlength,cell/.style={draw=blue!50,fill=blue!10,rectangle,minimum size=\boxlength},thick]
				\coordinate (label) at (-0.5, 0);
				\foreach \i in {0,1,4,6,7,10,12,13,16,18,19,22}{
					\node[cell,anchor=south west] (l\i) at (\i,0) {$\nomove$};
					\node[cell,anchor=north west,yshift=\pgflinewidth] (m\i) at (\i,0) {$t_\star$};
				}
				\foreach \i/\j in {5/1,11/2,17/q}{
					\node[cell,fill=green!20,anchor=south west] (l\i) at (\i,0) {$\sigma_\j$};
					\node[cell,fill=green!20,anchor=north west,yshift=\pgflinewidth] (m\i) at (\i,0) {$\star$};
					\node[anchor=south] (i\i) at (\i.5,1) {$i_\j\strut$};
					\node[anchor=north] (M\i) at (\i+1,-1) {$M_{\j}^{\phantom{'}}$};
				}
				\node[anchor=north] (M0) at (0,-1) {$M_0^{\phantom{'}}$};
				\foreach \i/\j in {5/0,11/1,17/{q-1},23/q}{
					\node[anchor=north] (Mp\i) at (\i,-1) {$M_{\j}^{'}$};
				}
				\foreach \i in {3,9,15,21}{
					\node (d\i) at (\i,0) {$\cdots$};
				}
				\node[anchor=south] (di) at (15,1) {$\cdots\strut$};
				\foreach \i/\j in {l/\sigma,m/S}{
					\node[anchor=base] (\i l) at (label |- \i0.base) {$\j$};
				}
				\begin{scope}[decoration={brace,amplitude=0.2\boxlength}]
					\draw[decorate] (-1,-1) -- node[left=.1\boxlength] {$\gamma\;\strut$} (-1,1);
					\foreach \i/\j/\k in {0/6/0,6/12/1,18/23/q}{
						\draw[decorate] (\i+0.1,2) -- node[above=.1\boxlength] {$\delta_\k\strut$} (\j-0.1,2);
						\coordinate (bl) at (\i.5, 0);
						\draw[blue!50,-Stealth,very thick,bend right=10] (bl|-M0.base) to +(4,0);
					}
				\end{scope}
			\end{tikzpicture}
			\caption{Decomposing the alignment $\gamma$ into parts $\delta_j$ at non-model move positions $i_j$.
				The proof idea is to shorten model move sequences (purple arrows) to connecting sequences of length at most $b\cdot\abs{T}\cdot(\abs{T}+1)\cdot(\abs{T}+2)/6$. $t_\star$ is a placeholder for any transition in $T$, $\star$ for a transition in $T$ or the no-move symbol $\nomove$.}
			\label{fig:splitting:gamma}
		\end{figure}
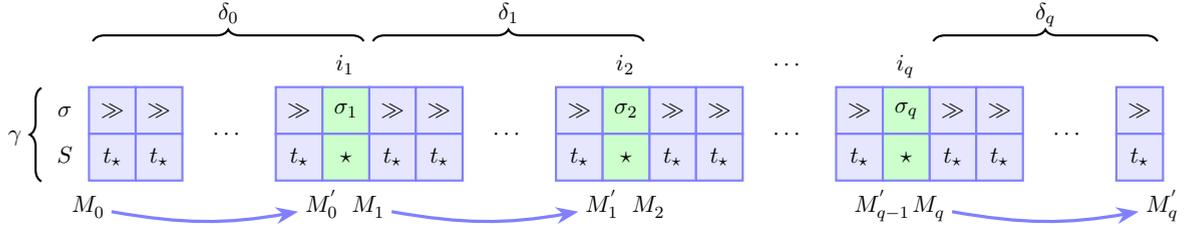
	\end{proof}
	
	\zcref[S]{thm:lbf:poly:alignments} yields an $\NP$-strategy for the alignment problem on \lbfc-systems:
	on input $S$, $\sigma$, and $k$, where $S$ is an accepting \lbfc-system, $\sigma$ is a trace as above, and $k \in \costRange$ denotes a threshold, the problem is to decide whether some alignment $\gamma \in \Gamma_{\sigma,S}$ exists with costs $\seqsum c(\gamma) \leq k$.
	We make use of \zcref{thm:lbf:poly:alignments} and non-deterministically construct an alignment $\gamma$ of length at most $(\abs\sigma + 1) \cdot (b \cdot \abs T \cdot (\abs T+1) \cdot (\abs T+2)/6 + 1)$ and verify that:
	\begin{itemize}
		\item $\gamma$ is a valid alignment between $\sigma$ and $S$, and
		\item its costs $\seqsum c(\gamma)$ do not exceed~$k$.
	\end{itemize}
	By \zcref{thm:lbf:poly:alignments}, an optimal alignment between $\sigma$ and $S$ is among the potential candidates.
	Note that (since $b$ is given in unary) the length of $\gamma$ is bounded polynomially in $\sigma$ and $S$. Moreover, it is easy to verify that $\gamma$ is a valid alignment by just simulating the system $S$ and $\sigma$ accordingly and to check that its costs do not exceed $k$. Hence, we obtain:
	\begin{corollary}
		\label{cor:lbf:np:membership}%
		On the class of \lbfc-systems, $\CPalign$ is in~$\NP$.
	\end{corollary}
	
	We can apply this result to sound free-choice workflow nets as well. By \zcref{lem:shortcircuit_align}, $\CPalign$ on sound free-choice workflow nets reduces to $\CPalign$ on their short-circuited nets, which are live and bounded, and still free-choice, since $t_\Lreset$ is the only output transition of $o$. Moreover, sound free-choice workflow nets are safe~\autocite[Lemma~1]{vanderAalst2000}, so \zcref{cor:lbf:np:membership} applies with $b=1$.
	Given their key role in the process mining field, let us state this result explicitly:
	\begin{corollary}
		\label{cor:fcwfnets:np:membership}%
		On the class of sound free-choice workflow nets, $\CPalign$ is in~$\NP$.
	\end{corollary}
	
	\section{Membership as a Lower Bound for Alignments}
	\label{sec:membership}
	
	Another closely related problem is the \emph{membership problem}. It determines whether a trace is part of the language of a given Petri net, i.e., whether it occurs as the labeling of a complete firing sequence of a given accepting system.
	\begin{problem}{Membership ($\CPmember$)}
		\emph{Input:} An alphabet $\Sigma$, an accepting system $S=(N,M_\init,M_\final)$ with $N=(P,T,F,\ell)$ and labeling function $\ell\colon T\to\Sigmatau$, and a trace $\sigma\in\Sigma^*$ over $\Sigma$.
		
		\emph{Question:} Is $\sigma\in\lang{S}$?
	\end{problem}
	
	We next observe that $\CPmember$ is a special case of $\CPalign$ where we look for a \emph{perfect} alignment, i.e., an alignment with costs $0$ with respect to the standard cost function.
	\begin{lemma}
		\label{lem:member_to_align}%
		For easy-sound accepting systems, $\CPmember$ is polynomial-time reducible to $\CPalign$.
	\end{lemma}
	\begin{proof}
		Let $\Sigma$, $S=(N,M_\init,M_\final)$, $N=(P,T,F,\ell)$, $\ell\colon T\to\Sigmatau$, and $\sigma\in\Sigma^*$ be an input of the $\CPmember$ problem.
		For the reduction to $\CPalign$ we make use of the standard cost function and set the threshold to $k = 0$.
		In effect, we are looking for a \emph{perfect} alignment between the trace $\sigma$ and the system $S$ (i.e., with costs~0 with respect to the standard cost function).
		Note that a perfect alignment can only consist of synchronous moves and silent model moves. Thus, such an alignment exists if and only if the trace $\sigma$ can be obtained as the labeling of a firing sequence of the system $S$ from $M_\init$ to $M_\final$.
		This is precisely the decision problem $\CPmember$.
	\end{proof}
	Similarly as in the proof of~\zcref{lem:mincostreach_to_align}, we have one technical problem: the system $S$ needs to be easy-sound.
	Again, we could readily solve this with a gadget adequate for the Petri net class at hand. For safe systems, for example, adding a transition $t_\Lskip$ with $\pset t_\Lskip=\supp{M_\init}$ and $t_\Lskip\pset=\supp{M_\final}$ with a new label not present in $\sigma$ would suffice. If this transition is taken in some alignment, we neither have a synchronous nor a silent move, and thus the alignment costs are at least $1$.
	In the following, however, we will apply this reduction to sound systems only, so that we can omit the gadget here completely and instead limit the input to systems that are already easy-sound.
	
	It remains a question for future research if a reduction in the other direction also holds, i.e., if $\CPalign$ is polynomial-time reducible to $\CPmember$.
	
	\section{Alignments on Process Trees}
	\label{sec:process_trees}
	
	Process trees are well-known models for \emph{block-structured} business processes. Intuitively, process trees constrain the control flow in a tree-like fashion, thereby excluding (long-term) dependencies between different branches of the model. More formally, process trees are a syntactic fragment of workflow nets which guarantees soundness and decomposability.
	The restricted expressive power leads to much better algorithmic behavior and prevents certain learning algorithms from overfitting (of course, by accepting a representational bias of the model). Indeed, one of the most successful process discovery algorithms, the \emph{Inductive Miner}~\autocite{Leemans2017,Leemans2022,LeemansFA2014}, uses process trees as its core. This is why process trees have drawn considerable attention in the area of process mining.
	
	Despite being generally accepted as \emph{simple} process models, we show that the alignment problem for process trees still remains hard, specifically $\NP$-complete.
	However, at the same time, our result witnesses that process trees do enjoy better algorithmic properties than safe and sound workflow nets for which the alignment problem is $\PSPACE$-complete.
	
	We proceed as follows. First, we derive $\NP$-membership from the fact that process trees translate into equivalent sound free-choice workflow nets of polynomial size.
	To establish $\NP$-hardness, we uncover a tight connection between process trees and \emph{shuffle languages}, a class of regular expressions that has been studied extensively in the area of formal language theory. Of course, this new connection may prove quite useful in other contexts as well.
	We start with a formal definition of shuffle languages and process trees.
	\begin{definition}[Shuffle $\shuffle$]
		For $x,y \in \Sigma^*$, the \emph{shuffle} $x \shuffle y$ of $x$ and $y$ is defined by
		\begin{equation*}
			x \shuffle y \coloneq \Set*{v_1 w_1 \dots v_k w_k \given x = v_1 \dots v_k, y = w_1 \dots w_k, v_i, w_i \in \Sigma^*, 1\leq i\leq k}.
		\end{equation*}
		Let $\mathcal L_1, \mathcal L_2 \subseteq \Sigma^*$ be two languages. The shuffle of $\mathcal L_1$ and $\mathcal L_2$ is defined by
		\begin{equation*}
			\mathcal L_1\shuffle\mathcal L_2 \coloneq \bigcup\Set{w_1\shuffle w_2\given w_1\in\mathcal L_1, w_2\in\mathcal L_2}.
		\end{equation*}
		For multiple words $w_1,\dots,w_n\in\Sigma^*$ or languages $\mathcal L_1,\dots,\mathcal L_n\subseteq\Sigma^*$, we also define
		\begin{equation*}
			\bigshuffle_{i=1}^n w_i\coloneq w_1\shuffle\dots\shuffle w_n\qquad\text{and}\qquad\bigshuffle_{i=1}^n\mathcal L_i\coloneq\mathcal L_1\shuffle\dots\shuffle\mathcal L_n.
		\end{equation*}
	\end{definition}
	
	The shuffle operator captures (independent) parallel executions. During our research on process trees we found that they have already been studied before in the area of formal language theory, see, e.g.,~\autocite{MayerS1994}. Here, they are known as \emph{shuffle languages} or \emph{regular expressions with interleaving}. Indeed, process trees are nothing more than usual regular expressions extended by the shuffle operator~$\shuffle$.
	\begin{definition}[Process Tree]
		\label{def:process_tree}%
		Let $\Sigma$ be an alphabet and let $\tau\notin\Sigma$ be the silent activity. Both, the set of \emph{process trees} (over $\Sigma$) and the \emph{language} of a process tree $T$, denoted by $\lang{T}$, are defined recursively:
		\begin{itemize}
			\item the silent activity $\tau$ is a process tree where $\lang{\tau}=\Set{\seq{}}$,
			\item each activity $a\in\Sigma$ is a process tree where $\lang{a}=\Set{\seq{a}}$,
			\item for process trees $T_1$, \dots, $T_n$, $n\in\N\setminus\Set{0}$, the following are also process trees:
			\begin{itemize}
				\item $\ptseq(T_1,\dots,T_n)$ where $\lang{\ptseq(T_1,\dots,T_n)}=\lang{T_1}\cdot\ldots\cdot\lang{T_n}$,
				\item $\ptxor(T_1,\dots,T_n)$ where $\lang{\ptxor(T_1,\dots,T_n)}=\lang{T_1}\cup\ldots\cup\lang{T_n}$,
				\item $\ptpar(T_1,\dots,T_n)$ where $\lang{\ptpar(T_1,\dots,T_n)}=\lang{T_1}\shuffle\ldots\shuffle\lang{T_n}$, and
				\item $\ptloop(T_1,T_2)$ where $\lang{\ptloop(T_1,T_2)}=\lang{T_1}\cdot(\lang{T_2}\cdot\lang{T_1})^*$.
			\end{itemize}
		\end{itemize}
		The symbols $\ptseq$ (sequence), $\ptxor$ (exclusive choice), $\ptpar$ (parallel), and $\ptloop$ (loop) are \emph{process tree operators}.
		A \emph{process tree with unique labels} is a process tree where each activity occurs at most once.
	\end{definition}
	
	It is not hard to see that all shuffle languages are regular (just translate the shuffle operator by using a standard product construction). Hence, also the behavior defined by a process tree is a regular language (which is, of course, not true for the \emph{iterated shuffle}).
	This also means that each process tree can be rewritten as an equivalent process tree (regular expression) that does not make use of the parallel operator ($\ptpar$). However, in general, this transformation inevitably leads to a process tree that is exponentially larger than the original one~\autocite{MayerS1994}. Interestingly, this phenomenon is well-known from other models with concurrency as, for example, from safe Petri nets.
	
	Our goal is to classify the algorithmic complexity of the alignment problem for process trees. To this end we make use of the intimate connection between process trees and shuffle languages. In fact, the computational complexity of important decision problems for shuffle languages is well-known. Let us formulate the following problems and results for later reference.
	\begin{problem}{Shuffle of Words ($\CPwordshuffle$)}
		\emph{Input:} An alphabet $\Sigma$, and words $v\in\Sigma^*$ and $w_1,\dots,w_n\in\Sigma^*$ over $\Sigma$.
		
		\emph{Question:} Is $v\in\bigshuffle_{i=1}^n w_i$?
	\end{problem}
	\begin{problem}{Membership for Shuffle Languages ($\CPmembershuffle$)}
		\emph{Input:} An alphabet $\Sigma$, a regular expression $R$ over $\Sigma$ using the operators $\cdot$ (concatenation), $\shuffle$ (shuffle), $^*$ (Kleene star), $|$ (union), $\cap$ (intersection), and a word $w\in\Sigma^*$ over $\Sigma$.
		
		\emph{Question:} Is $w \in \lang R$?
	\end{problem}
	$\CPwordshuffle$ and $\CPmembershuffle$ are known to be $\NP$-complete~\autocite{Mansfield1983,WarmuthH1984}.
	As an immediate consequence we get:
	\begin{proposition}
		\label{prop:pt_member}%
		On the class of process trees, $\CPmember$ is $\NP$-complete, even if restricted to the sequence ($\ptseq$) and parallel ($\ptpar$) operator.
	\end{proposition}
	\begin{proof}
		Since $\CPwordshuffle$ is $\NP$-complete, the membership problem $\CPmember$ for process trees (i.e., regular expressions with shuffle operator) is $\NP$-hard even for process trees built only from concatenation and shuffle, i.e., the process tree operators sequence ($\ptseq$) and parallel ($\ptpar$).
		Moreover, as $\CPmembershuffle$ is $\NP$-complete, the complexity of $\CPmember$ does not increase if we throw in the remaining operators exclusive choice ($\ptxor$) and loop ($\ptloop$).
	\end{proof}
	
	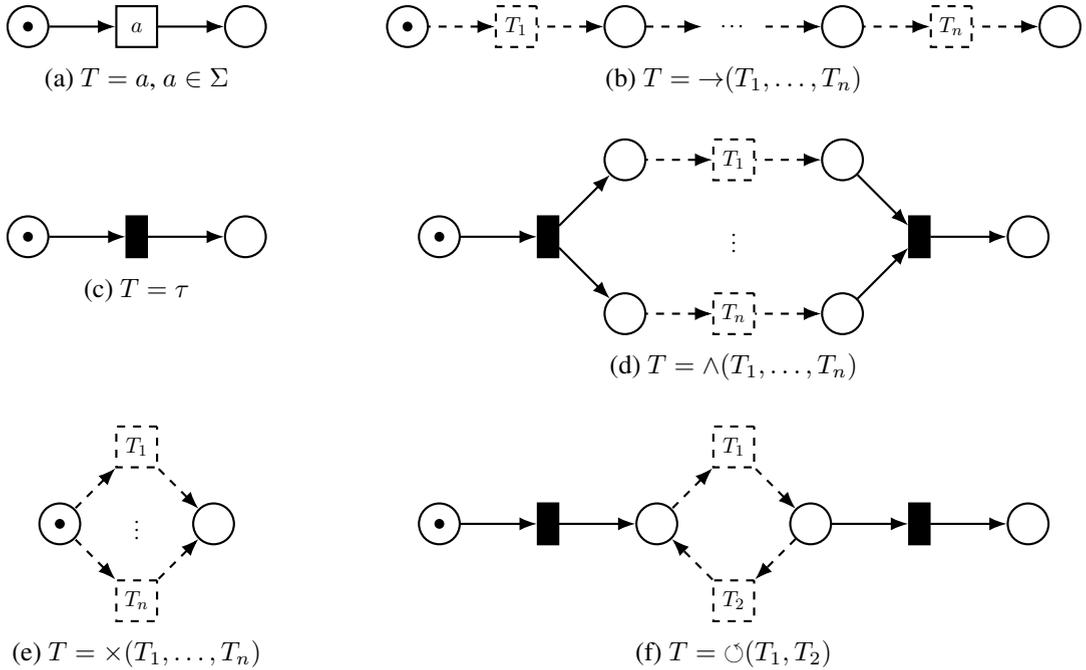
\begin{figure}
		\newdimen\intervspace
		\intervspace=2ex
		\tikzset{every picture/.style={scale=.75, transform shape}}
		\begin{subfigure}{.33\linewidth}
			\centering
			\begin{tikzpicture}[on grid]
				\node[place,          tokens=1]                     (pi)    {};
				\node[transition,     position=  0 degrees from pi] (t1) {$a$} edge[pre] (pi);
				\node[place,          position=  0 degrees from t1] (pf)    {} edge[pre] (t1);
			\end{tikzpicture}
			\caption{$T=a$, $a\in\Sigma$}
		\end{subfigure}
		\begin{subfigure}{.67\linewidth}
			\centering
			\begin{tikzpicture}[on grid]
				\node[place,          tokens=1]                             (pi)                       {};
				\node[transition,     position=  0 degrees from pi, dashed] (t1)                  {$T_1$} edge[pre, dashed]  (pi);
				\node[place,          position=  0 degrees from t1]         (p1)                       {} edge[pre, dashed]  (t1);
				\node[                position=  0 degrees from p1]         (t2)  {\enspace\dots\enspace} edge[pre, dashed]  (p1);
				\node[place,          position=  0 degrees from t2]         (p2)                       {} edge[pre, dashed]  (t2);
				\node[transition,     position=  0 degrees from p2, dashed] (t3)                  {$T_n$} edge[pre, dashed]  (p2);
				\node[place,          position=  0 degrees from t3]         (pf)                       {} edge[pre, dashed]  (t3);
			\end{tikzpicture}
			\caption{$T=\ptseq(T_1,\dots,T_n)$}
		\end{subfigure}
		
		\vspace{\intervspace}
		
		\begin{minipage}[c]{.33\linewidth}
			\begin{subfigure}{\linewidth}
				\centering
				\begin{tikzpicture}[on grid]
					\node[place,          tokens=1]                     (pi)    {};
					\node[tau-transition, position=  0 degrees from pi] (t1)    {} edge[pre] (pi);
					\node[place,          position=  0 degrees from t1] (pf)    {} edge[pre] (t1);
				\end{tikzpicture}
				\caption{$T=\tau$}
			\end{subfigure}
		\end{minipage}
		\begin{minipage}[c]{.67\linewidth}
			\begin{subfigure}{\linewidth}
				\centering
				\begin{tikzpicture}[on grid]
					\node[place,          tokens=1]                             (pi)                       {};
					\node[tau-transition, position=  0 degrees from pi]         (t1)                       {} edge[pre]          (pi);
					\node[place,          position= 45 degrees from t1]         (p1)                       {} edge[pre]          (t1);
					\node[place,          position=-45 degrees from t1]         (p2)                       {} edge[pre]          (t1);
					\node[transition,     position=  0 degrees from p1, dashed] (t2)                  {$T_1$} edge[pre, dashed]  (p1);
					\node[transition,     position=  0 degrees from p2, dashed] (t4)                  {$T_n$} edge[pre, dashed]  (p2);
					\node[                between=t2 and t4]                    (t3) {\enspace\vdots\enspace};
					\node[place,          position=  0 degrees from t2]         (p3)                       {} edge[pre, dashed]  (t2);
					\node[place,          position=  0 degrees from t4]         (p4)                       {} edge[pre, dashed]  (t4);
					\node[tau-transition, position= 45 degrees from p4]         (t5)                       {} edge[pre]          (p3) edge[pre]          (p4);
					\node[place,          position=  0 degrees from t5]         (pf)                       {} edge[pre]          (t5);
				\end{tikzpicture}
				\caption{$T=\ptpar(T_1,\dots,T_n)$}
			\end{subfigure}
		\end{minipage}
		
		\vspace{\intervspace}
		
		\begin{subfigure}{.33\linewidth}
			\centering
			\begin{tikzpicture}[on grid]
				\node[place,          tokens=1]                             (pi)                       {};
				\node[transition,     position= 45 degrees from pi, dashed] (t1)                  {$T_1$} edge[pre, dashed]  (pi);
				\node[transition,     position=-45 degrees from pi, dashed] (t3)                  {$T_n$} edge[pre, dashed]  (pi);
				\node[                between=t1 and t3]                    (t2) {\enspace\vdots\enspace};
				\node[place,          position= 45 degrees from t3]         (pf)                       {} edge[pre, dashed]  (t1) edge[pre, dashed]  (t3);
			\end{tikzpicture}
			\caption{$T=\ptxor(T_1,\dots,T_n)$}
		\end{subfigure}
		\begin{subfigure}{.67\linewidth}
			\centering
			\begin{tikzpicture}[on grid]
				\node[place,          tokens=1]                             (pi)                       {};
				\node[tau-transition, position=  0 degrees from pi]         (t1)                       {} edge[pre]          (pi);
				\node[place,          position=  0 degrees from t1]         (p1)                       {} edge[pre]          (t1);
				\node[transition,     position= 45 degrees from p1, dashed] (t2)                  {$T_1$} edge[pre, dashed]  (p1);
				\node[transition,     position=-45 degrees from p1, dashed] (t3)                  {$T_2$} edge[post, dashed] (p1);
				\node[place,          position= 45 degrees from t3]         (p2)                       {} edge[pre, dashed]  (t2) edge[post, dashed] (t3);
				\node[tau-transition, position= 45 degrees from p4]         (t4)                       {} edge[pre]          (p2);
				\node[place,          position=  0 degrees from t4]         (pf)                       {} edge[pre]          (t4);
			\end{tikzpicture}
			\caption{$T=\ptloop(T_1,T_2)$}
		\end{subfigure}
		\caption{Workflow net representations of different process trees $T$.}
		\label{fig:pt_transformation}
	\end{figure}
	From this result, it is not hard to infer that $\CPalign$ is $\NP$-complete on process trees as well.
	There is a small caveat: we defined alignments between traces and \emph{Petri nets}, while we specified the semantics of process trees without reference to a Petri net representation.
	This is a minor problem, however, since process trees can be translated into equivalent safe and sound workflow nets (and we can use the standard cost function for the alignment computation). An explicit construction can be found in~\autocite{Leemans2017,Leemans2022}, which is also illustrated in \zcref{fig:pt_transformation}.
	\begin{lemma}
		\label{lem:pt_to_wfnet}%
		Using the construction from \zcref{fig:pt_transformation}, a process tree can be transformed into an equivalent safe and sound free-choice workflow net in polynomial time.
	\end{lemma}
	
	This allows us to now formulate the central result of this section:
	\begin{theorem}
		\label{thm:pt_npcomplete}%
		On the class of process trees, $\CPalign$ is $\NP$-complete.
	\end{theorem}
	\begin{proof}
		Membership in $\NP$ follows from the fact that process trees efficiently translate into safe and sound free-choice workflow nets (\zcref{lem:pt_to_wfnet}).
		For the same reason, we can directly apply the reduction in \zcref{lem:member_to_align}, so that $\CPmember$ is a special case of $\CPalign$. Then, $\NP$-hardness follows from \zcref{prop:pt_member}.
	\end{proof}
	
	With \zcref{lem:pt_to_wfnet}, we can also transfer the result of \zcref{thm:pt_npcomplete} to sound free-choice workflow nets, and together with \zcref{cor:fcwfnets:np:membership}, we obtain:
	\begin{corollary}
		\label{cor:align_sound_fc_wfnets_npcomplete}%
		On the class of sound free-choice workflow nets, $\CPalign$ is $\NP$-complete.
	\end{corollary}
	By \zcref{lem:shortcircuit_align}, short-circuiting the resulting workflow nets extends the $\NP$-hardness also to cyclic and safe \lbfc-systems, since $t_\Lreset$ is the only output transition of $o$. Together with \zcref{cor:lbf:np:membership}, we obtain:
	\begin{corollary}
		\label{cor:align_lfbc_npcomplete}%
		On the class of \lbfc-systems, $\CPalign$ is $\NP$-complete. $\NP$-hardness holds even for cyclic and safe \lbfc-systems.
	\end{corollary}
	Furthermore, these results also apply to \emph{Partially Ordered Workflow Language} (POWL) models~\autocite{KouraniZ2023}, a generalization of process trees which additionally allows partial orders over submodels. As POWL models can still be translated into equivalent sound free-choice workflow nets in polynomial time~\autocite{KouraniZSA2025}, we have:
	\begin{corollary}
		\label{cor:align_powl_npcomplete}%
		On the class of POWL models, $\CPalign$ is $\NP$-complete.
	\end{corollary}
	
	In our previous work~\autocite{SchwanenPA2025}, we exploited the $\NP$-membership of $\CPalign$ and provided the first mixed integer linear program encoding of alignments on process trees. Our experiments show that this approach can lead to significant speed-ups compared to the standard $A^*$-based algorithm.

	Another important subclass of process trees arises in the scope of the Inductive Miner~\autocite{Leemans2017,Leemans2022,LeemansFA2014}. This widely-used mining algorithm yields \emph{process trees with unique labels} which means that each activity label occurs at most once in the tree.
	We recently showed that $\CPalign$ can be solved in polynomial time on process trees with unique labels \autocite{SchwanenPA2025a} using a dynamic programming approach.
	\begin{theorem}[{{\normalfont\autocite[Theorem~4]{SchwanenPA2025a}}}]
		\label{thm:pt:unique}%
		On the class of process trees with unique labels, $\CPalign$ is in $\P$.
	\end{theorem}
	As a direct result of our findings here and in \autocite{SchwanenPA2025,SchwanenPA2025a,SchwanenPA2025c} we can conclude that $\CPalign$ on process trees stays within $\P$ as long as each subtree rooted at a parallel operator does not contain any duplicate activity labels.
	
	The $\NP$-completeness of $\CPwordshuffle$ implies that a basic constructor for concurrency, such as the shuffle operator, already leads to $\NP$-hardness of $\CPmember$ (and hence of $\CPalign$).
	At the same time, even the most basic process modeling notations provide such operators. To see this, note that a \emph{trace system} (cf.\@ \zcref{def:trace_system}) models a single trace $w$, so the parallel composition of $n$ trace systems represents the shuffle of $n$ traces $w_1,\dots,w_n$. At the same time, this union of trace systems surely is syntactically an extremely simple formalism.
	Thus, if we want to break the $\NP$ barrier and obtain classes of process models which allow for \emph{efficient} alignment computation, we have to exclude arbitrary mixtures of choice and synchronization.
	We take a closer look at this finding in the following section.
	
	\section{Alignments on S-Systems and T-Systems}
	\label{sec:stsystems}
	
	We proceed with further analyzing the source of $\NP$-hardness by focusing on the two primitive workflow mechanisms \enquote{\emph{choice}} and \enquote{\emph{concurrency}}. While $\CPreach$ is known to be in $\P$ on the classes of S-systems (a model for systems with choices, but without synchronization) and T-systems (a dual class of Petri nets without choices, but with synchronization) \autocite{BestT1987,CommonerHEP1971,GenrichL1973}, we show that on these classes, the complexity of $\CPalign$ depends on the number of tokens that can be present simultaneously: if it is bounded by a constant, $\CPalign$ is solvable in polynomial time, whereas it is $\NP$-hard as soon as this number is part of the input, even for safe S-systems and for safe T-systems with a single initial token.
	
	A definition of T-nets and T-systems is already given in \zcref{def:tsystem}. From the previous section, we know that $\CPwordshuffle$ and therefore $\CPalign$ on process trees is $\NP$-complete, even if the process tree is restricted to the sequence ($\ptseq$) and parallel ($\ptpar$) operator (cf.\@ \zcref{prop:pt_member}). Given such a process tree and using the construction from \zcref{fig:pt_transformation}, we immediately see that all introduced places have at most one incoming and one outgoing arc; thus, the resulting workflow net is a safe and sound T-system and even acyclic. Therefore, we can conclude:
	\begin{corollary}
		\label{cor:align_tsystem_nphard}%
		On the classes of T-systems and of acyclic systems, $\CPalign$ is $\NP$-hard. $\NP$-hardness holds even for T-systems that are safe, sound, and acyclic and for acyclic systems that are safe and sound.
	\end{corollary}
	Note that the resulting workflow nets start with a single token. The number of tokens that can be present simultaneously is only created by the silent transitions splitting the parallel branches, and it grows with the number of words in the $\CPwordshuffle$ instance.
	
	We observe concurrency whenever two (or more) transitions of the Petri net can fire independently, i.e., firing one transition does not disable the other transition. Note that this is exactly the situation as in \emph{persistent} or \emph{conflict-free} systems. This is of course no surprise as T-systems are indeed conflict-free and thus persistent.
	\begin{corollary}
		\label{cor:align_safe_conflictfree_nphard}%
		On the classes of persistent systems and of conflict-free systems, $\CPalign$ is $\NP$-hard. $\NP$-hardness holds even for persistent systems and conflict-free systems that are safe, sound, and acyclic.
	\end{corollary}
	
	We can now already show that $\CPalign$ is $\NP$-complete for bounded (or even safe) T-systems, even under additional behavioral assumptions.
	\begin{theorem}
		\label{thm:align_bounded_tsystem_np}%
		On the class of bounded T-systems, $\CPalign$ is $\NP$-complete. $\NP$-hardness holds even for safe and sound acyclic T-systems and for live and safe T-systems.
	\end{theorem}
	\begin{proof}
		Since all firing sequences in T-systems are biased, \zcref{lem:lbf:biased:sequence} directly applies to arbitrary firing sequences of $b$-bounded T-systems, independent of their initial marking. By replacing \zcref{thm:shortestst:tsystems:rev} in the proof of \zcref{thm:lbf:poly:alignments} with \zcref{lem:lbf:biased:sequence}, membership in $\NP$ follows.
		For $\NP$-hardness, the workflow nets constructed for \zcref{cor:align_tsystem_nphard} are safe and sound acyclic T-systems. Since they are acyclic, their source place $i$ has no input transition, so short-circuiting them (\zcref{lem:shortcircuit_align}) preserves the T-net property and yields cyclic, live, safe T-systems without affecting the reduction.
	\end{proof}
	
	\begin{figure}[t]
		\begin{subfigure}{\linewidth}
			\centering
			\begin{tikzpicture}[scale=.49, transform shape, on grid]
				\node[place, tokens=1]              (pi)  {};
				\node[right of=pi,  tau-transition] (ts)  {}  edge[pre] (pi);
				\node[right of=ts,  place]          (p20) {}  edge[pre] (ts);
				\node[right of=p20, transition]     (t21) {T} edge[pre] (p20);
				\node[right of=t21, place]          (p21) {}  edge[pre] (t21);
				\node[right of=p21, transition]     (t22) {U} edge[pre] (p21);
				\node[right of=t22, place]          (p22) {}  edge[pre] (t22);
				\node[right of=p22, transition]     (t23) {R} edge[pre] (p22);
				\node[right of=t23, place]          (p23) {}  edge[pre] (t23);
				\node[right of=p23, transition]     (t24) {I} edge[pre] (p23);
				\node[right of=t24, place]          (p24) {}  edge[pre] (t24);
				\node[right of=p24, transition]     (t25) {N} edge[pre] (p24);
				\node[right of=t25, place]          (p25) {}  edge[pre] (t25);
				\node[right of=p25, transition]     (t26) {G} edge[pre] (p25);
				\node[right of=t26, place]          (p26) {}  edge[pre] (t26);
				\node[above of=t21, place]          (p10) {}  edge[pre] (ts);
				\node[right of=p10, transition]     (t11) {P} edge[pre] (p10);
				\node[right of=t11, place]          (p11) {}  edge[pre] (t11);
				\node[right of=p11, transition]     (t12) {E} edge[pre] (p11);
				\node[right of=t12, place]          (p12) {}  edge[pre] (t12);
				\node[right of=p12, transition]     (t13) {T} edge[pre] (p12);
				\node[right of=t13, place]          (p13) {}  edge[pre] (t13);
				\node[right of=p13, transition]     (t14) {R} edge[pre] (p13);
				\node[right of=t14, place]          (p14) {}  edge[pre] (t14);
				\node[right of=p14, transition]     (t15) {I} edge[pre] (p14);
				\node[right of=t15, place]          (p15) {}  edge[pre] (t15);
				\node[below of=p20, place]          (p30) {}  edge[pre] (ts);
				\node[right of=p30, transition]     (t31) {G} edge[pre] (p30);
				\node[right of=t31, place]          (p31) {}  edge[pre] (t31);
				\node[right of=p31, transition]     (t32) {O} edge[pre] (p31);
				\node[right of=t32, place]          (p32) {}  edge[pre] (t32);
				\node[right of=p32, transition]     (t33) {E} edge[pre] (p32);
				\node[right of=t33, place]          (p33) {}  edge[pre] (t33);
				\node[right of=p33, transition]     (t34) {D} edge[pre] (p33);
				\node[right of=t34, place]          (p34) {}  edge[pre] (t34);
				\node[right of=p34, transition]     (t35) {E} edge[pre] (p34);
				\node[right of=t35, place]          (p35) {}  edge[pre] (t35);
				\node[right of=p35, transition]     (t36) {L} edge[pre] (p35);
				\node[right of=t36, place]          (p36) {}  edge[pre] (t36);
				\node[right of=p26, tau-transition] (te)  {}  edge[pre] (p15) edge[pre] (p26) edge[pre] (p36);
				\node[right of=te,  place]          (pf)  {}  edge[pre] (te);
			\end{tikzpicture}
			\caption{System which is constructed based on a process tree with the language $\lang{w_1\shuffle w_2\shuffle w_3}$ and therefore results in a safe and sound free-choice workflow net which is even acyclic and a T-system, and thus conflict-free.}\label{fig:tsystem:np}
		\end{subfigure}
		
		\vspace{1em}
		
		\begin{subfigure}{\linewidth}
			\centering
			\begin{tikzpicture}[scale=.49, transform shape, on grid]
				\node[place, tokens=1]              (pi)  {};
				\node[right of=pi,  tau-transition] (ts)  {}  edge[pre] (pi);
				\node[right of=ts,  place]          (p20) {}  edge[pre] (ts);
				\node[right of=p20, transition]     (t21) {P} edge[pre] (p20);
				\node[right of=t21, place]          (p21) {}  edge[pre] (t21);
				\node[right of=p21, transition]     (t22) {E} edge[pre] (p21);
				\node[right of=t22, place]          (p22) {}  edge[pre] (t22);
				\node[right of=p22, transition]     (t23) {T} edge[pre] (p22);
				\node[right of=t23, place]          (p23) {}  edge[pre] (t23);
				\node[right of=p23, transition]     (t24) {R} edge[pre] (p23);
				\node[right of=t24, place]          (p24) {}  edge[pre] (t24);
				\node[right of=p24, transition]     (t25) {I} edge[pre] (p24);
				\node[right of=t25, place]          (p25) {}  edge[pre] (t25);
				\node[above of=p20, place]          (p10) {}  edge[pre] (ts);
				\node[right of=p10, transition]     (t11) {P} edge[pre] (p10);
				\node[right of=t11, place]          (p11) {}  edge[pre] (t11);
				\node[right of=p11, transition]     (t12) {E} edge[pre] (p11);
				\node[right of=t12, place]          (p12) {}  edge[pre] (t12);
				\node[right of=p12, transition]     (t13) {T} edge[pre] (p12);
				\node[right of=t13, place]          (p13) {}  edge[pre] (t13);
				\node[right of=p13, transition]     (t14) {R} edge[pre] (p13);
				\node[right of=t14, place]          (p14) {}  edge[pre] (t14);
				\node[right of=p14, transition]     (t15) {I} edge[pre] (p14);
				\node[right of=t15, place]          (p15) {}  edge[pre] (t15);
				\node[below of=p20, place]          (p30) {}  edge[pre] (ts);
				\node[right of=p30, transition]     (t31) {P} edge[pre] (p30);
				\node[right of=t31, place]          (p31) {}  edge[pre] (t31);
				\node[right of=p31, transition]     (t32) {E} edge[pre] (p31);
				\node[right of=t32, place]          (p32) {}  edge[pre] (t32);
				\node[right of=p32, transition]     (t33) {T} edge[pre] (p32);
				\node[right of=t33, place]          (p33) {}  edge[pre] (t33);
				\node[right of=p33, transition]     (t34) {R} edge[pre] (p33);
				\node[right of=t34, place]          (p34) {}  edge[pre] (t34);
				\node[right of=p34, transition]     (t35) {I} edge[pre] (p34);
				\node[right of=t35, place]          (p35) {}  edge[pre] (t35);
				\node[right of=p25, tau-transition] (te)  {}  edge[pre] (p15) edge[pre] (p25) edge[pre] (p35);
				\node[right of=te,  place]          (pf)  {}  edge[pre] (te);
			\end{tikzpicture}
			\caption{System which is constructed based on a process tree with the language $\lang{w_1\shuffle w_1\shuffle w_1}=\lang{\bigshuffle_{i=1}^3 w_1}$ and therefore results in a safe and sound free-choice workflow net which is even acyclic and a T-system, and thus conflict-free.}\label{fig:tsystem:np:same}
		\end{subfigure}
		
		\vspace{1em}
		
		\begin{subfigure}{\linewidth}
			\centering
			\begin{tikzpicture}[scale=.49, transform shape, on grid]
				\node[place, tokens=3]               (p10) {};
				\node[right of=p10, transition]      (t11) {P} edge[pre] (p10);
				\node[right of=t11, place]           (p11) {}  edge[pre] (t11);
				\node[right of=p11, transition]      (t12) {E} edge[pre] (p11);
				\node[right of=t12, place]           (p12) {}  edge[pre] (t12);
				\node[right of=p12, transition]      (t13) {T} edge[pre] (p12);
				\node[right of=t13, place]           (p13) {}  edge[pre] (t13);
				\node[right of=p13, transition]      (t14) {R} edge[pre] (p13);
				\node[right of=t14, place]           (p14) {}  edge[pre] (t14);
				\node[right of=p14, transition]      (t15) {I} edge[pre] (p14);
				\node[right of=t15, place]           (p15) {}  edge[pre] (t15);
			\end{tikzpicture}
			\caption{Reduction to a single trace system $\mathcal T(w_1)$, but with multiple tokens in the initial place which results in a bounded S-system and T-system with the language $\lang{\bigshuffle_{i=1}^3 w_1}$.}\label{fig:stsystem:np}
		\end{subfigure}
		\caption{Example of safe T-systems and a bounded S-system composed of trace systems of the words $w_1=\text{PETRI}$, $w_2=\text{TURING}$, and $w_3=\text{GOEDEL}$.}
	\end{figure}
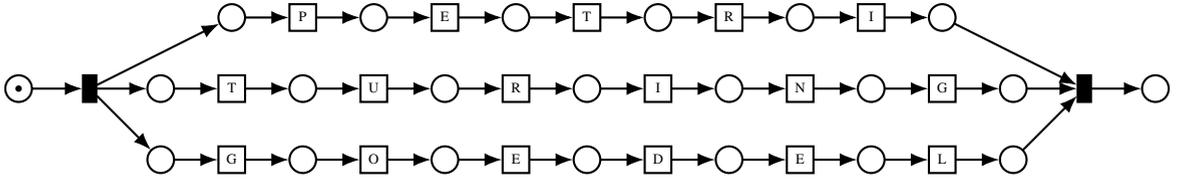
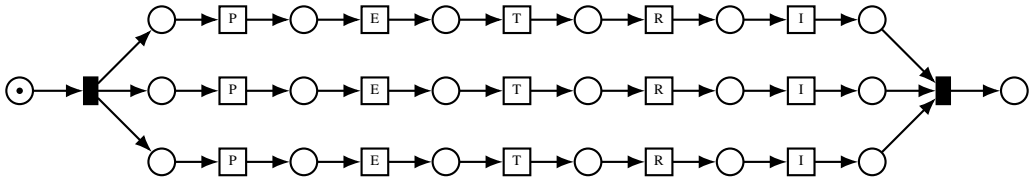
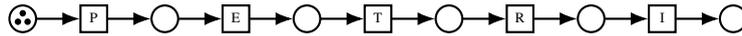
	As an example, let $w_1=\text{PETRI}$, $w_2=\text{TURING}$, and $w_3=\text{GOEDEL}$ be three words and we ask whether some input word $v$ is in $\lang{w_1\shuffle w_2\shuffle w_3}$. \zcref[S]{fig:tsystem:np} shows the safe T-system $S$ constructed according to \zcref{fig:pt_transformation} with $\lang{S}=\lang{w_1\shuffle w_2\shuffle w_3}$.
	It is also known that $\CPwordshuffle$ remains $\NP$-hard even when the words in the shuffle are identical~\autocite{WarmuthH1984}.
	\zcref[S]{fig:tsystem:np:same} shows the safe T-system $S$ constructed according to \zcref{fig:pt_transformation} with $\lang{S}=\lang{w_1\shuffle w_1\shuffle w_1}$. While the T-system $S$ in \zcref{fig:tsystem:np:same} also fulfills the properties of a workflow net, the identical parallel branches could also be merged into a single trace system. Then, the silent transitions can be omitted and the initial marking now consists of one token for each original parallel branch in the initial place. In fact, the resulting system shown in \zcref{fig:stsystem:np} is not only a T-system\textemdash although no longer safe\textemdash, but also a bounded S-system.
	\begin{definition}[S-Net, S-System]
		\label{def:ssystem}%
		A Petri net $N=(P,T,F,\ell)$ is an \emph{S-net} if and only if each transition has exactly one input and one output place, i.e., $\forall t\in T\colon\abs{\pset t}=1,\abs{t\pset}=1$. A system $(N,M_0)$ is an \emph{S-system} if $N$ is an S-net.
	\end{definition}
	
	Although firing sequences for S-systems are no longer biased, they are known to be short.
	\begin{theorem}[Shortest Sequence Theorem {\normalfont\autocite[Theorem~3.11]{DeselE1995}}]
		Let $(N,M_0)$ be a $b$-bounded S-system with $n$ transitions and let $M$ be a reachable marking.
		Then, there is a firing sequence $\sigma$ such that $(N,M_0)\fire\sigma(N,M)$ and
		\begin{equation*}
			\abs\sigma \leq b\cdot n.
		\end{equation*}
	\end{theorem}
	Again, the theorem only guarantees the \emph{existence} of a \enquote{short} firing sequence, which does not necessarily qualify as a \enquote{cheap} firing sequence. The proof of \autocite[Theorem~3.11]{DeselE1995}, however, already supports the generalized form:
	\begin{theorem}[Shortest Sequence Theorem---Generalized Form]
		\label{thm:ssystem_shortest_sequence}%
		Let $(N,M_0)$ with $N=(P,T,F,\ell)$ be an S-system with bound $b$.
		Moreover, let $\sigma\in T^*$ and $M_1,M_2\in\reach{N,M_0}$ such that $(N,M_1)\fire\sigma(N,M_2)$.
		Then, there exists a firing sequence $\sigma'\in T^*$ with $\vec{\sigma'}\leq\vec{\sigma}$ such that $(N,M_1)\fire{\sigma'}(N,M_2)$ and
		\begin{equation*}
			\abs{\sigma'} \leq b\cdot\abs{T}.
		\end{equation*}
	\end{theorem}
	With this, we obtain the counterpart of \zcref{thm:align_bounded_tsystem_np} for S-systems:
	\begin{theorem}
		\label{thm:align_bounded_ssystem_np}%
		On the class of bounded S-systems, $\CPalign$ is $\NP$-complete. $\NP$-hardness holds even for sound bounded acyclic S-systems and for live bounded S-systems.
	\end{theorem}
	\begin{proof}
		By replacing \zcref{thm:shortestst:tsystems:rev} in the proof of \zcref{thm:lbf:poly:alignments} with \zcref{thm:ssystem_shortest_sequence}, membership in $\NP$ follows.
		For $\NP$-hardness, consider the accepting system from \zcref{fig:stsystem:np} with the language $\bigshuffle_{j=1}^n w$ which is an acyclic, sound, and $n$-bounded S-system. As $\CPwordshuffle$ remains $\NP$-hard for identical words~\autocite{WarmuthH1984}, the reduction from \zcref{lem:member_to_align} yields $\NP$-hardness.
		For live bounded S-systems, we extend the net with a transition $t_\Lreset$ with $\pset t_\Lreset=\Set{p_{\abs{w}}}$ and $t_\Lreset\pset=\Set{p_0}$ and assign costs of $k+1$ to it (cf.\@ \zcref{lem:shortcircuit_align}). The resulting S-net is strongly connected and marked, hence live, and it is still $n$-bounded. As in the proof of \zcref{lem:shortcircuit_align}, no alignment with costs of at most $k$ uses $t_\Lreset$, so the reduction is not affected.
	\end{proof}
	Note that the S-systems used in this proof are bounded but not safe, and their bound $n$ grows with the input (since $n\leq\abs{v}$, it is still polynomial in unary notation). For live S-systems, this is unavoidable: as we will see after \zcref{prop:align_k_tokens_p}, $\CPalign$ is in $\P$ on live $b$-bounded S-systems for every constant $b$. If we instead keep safeness but drop liveness and soundness, the alignment problem remains hard:
	\begin{theorem}
		\label{thm:align_safe_ssystem_np}%
		On the class of safe S-systems, $\CPalign$ is $\NP$-complete.
	\end{theorem}
	\begin{proof}
		Membership follows from \zcref{thm:align_bounded_ssystem_np}.
		To show $\NP$-hardness, we describe a simple polynomial-time reduction from $\CPwordshuffle$. Let $v,w_1,\dots,w_n\in\Sigma^*$ denote an input for $\CPwordshuffle$. We take the disjoint union of the trace systems $\mathcal T(w_1),\dots,\mathcal T(w_n)$ with places $p_{i,0},\dots,p_{i,\abs{w_i}}$ for $i\in\Set{1,\dots,n}$. Thus, we obtain the initial marking $M_\init\coloneq\Multiset*{p_{i,0}\given 1\leq i\leq n}$ and the final marking $M_\final\coloneq\Multiset*{p_{i,\abs{w_i}}\given 1\leq i\leq n}$. One technicality remains: The net is not yet weakly connected. Therefore, we add a place $p_d$ and for each $i\in\Set{1,\dots,n}$ a transition $t_i^\Lrej$ with $\pset t_i^\Lrej=\Set{p_d}$ and $t_i^\Lrej\pset=\Set{p_{i,0}}$. Since $p_d$ is never marked, all transitions $t_i^\Lrej$ are dead. Hence, we obtain a safe S-system with the language $\bigshuffle_{i=1}^n w_i$, and \zcref{lem:member_to_align} completes the reduction.
	\end{proof}
	The dead transitions $t_i^\Lrej$ are essential for this construction. The next lemma shows that liveness or soundness together with safeness already forces a single token.
	\begin{lemma}
		\label{lem:safe_ssystem_single_token}%
		Let $S=(N,M_\init,M_\final)$ be a safe S-system with $M_\init\neq\emptyset$ such that $(N,M_\init)$ is live or $S$ is sound. Then, $M_\init$ consists of exactly one token.
	\end{lemma}
	\begin{proof}
		In both cases, $N$ has no dead transitions. For a place $p$, let $R(p)$ denote the set of places reachable from $p$ in the graph of $N$. Assume that two distinct places $p$ and $q$ are marked at $M_\init$. If there were a place $r\in R(p)\cap R(q)$, we could first move the token from $p$ to $r$ and then the token from $q$ to $r$, since each step only requires the token in the single input place. This would produce two tokens in some place, contradicting safeness. Hence, $R(p)\cap R(q)=\emptyset$, and no place outside of $\bigcup_{p\in\supp{M_\init}}R(p)$ is ever marked. Since no transition is dead, the input place of every transition lies in $R(p)$ for some $p\in\supp{M_\init}$, and then its output place lies in $R(p)$ as well. Consequently, no transition connects $R(p)$ and $R(q)$ for distinct $p,q\in\supp{M_\init}$, and all other places are isolated. This contradicts the weak connectivity of $N$.
	\end{proof}
	
	Due to $\NP$-hardness of $\CPalign$ on S-systems whose number of tokens is part of the input (\zcref{thm:align_bounded_ssystem_np,thm:align_safe_ssystem_np}), in the following, we bound the number of tokens by a constant, starting with a single token. By \zcref{lem:safe_ssystem_single_token}, a single token is equivalent to safeness for live or sound S-systems. This way we ensure that the system does not allow for any concurrency, which causes the construction of the reachability graph to be intractable for general systems. In S-systems, however, the number of tokens does not change when firing transitions and therefore, S-systems restricted to a single token directly translate to their reachability graph.
	\begin{definition}[Reachability Graph]
		\label{def:reachability_graph}%
		Let $(N,M_0)$ with $N=(P,T,F,\ell)$ be a system.
		Its \emph{reachability graph} $\reachGraph{N,M_0}\coloneq(\M,M_0,T,A,\ell)$ is a rooted directed graph with the set of reachable markings $\M\coloneq\reach{N,M_0}$ as set of vertices, the initial marking $M_0$ of the system as root, and a set of arcs $A\coloneq\Set{(M,t,M')\in\reach{N,M_0}\times T\times\reach{N,M_0}\given(N,M)\fire t(N,M')}$.
	\end{definition}
	Hence, their reachability graphs are known to be polynomial-time constructible.
	\begin{proposition}
		\label{prop:reachability_graph_p}%
		The reachability graph of an S-system where the initial marking consists of a single token can be constructed in polynomial time.
	\end{proposition}
	\begin{corollary}
		\label{cor:reachability_graph_p}%
		The reachability graph of a trace system is polynomial-time constructible.
	\end{corollary}
	
	Based on \zcref{prop:reachability_graph_p}, we show that the complexity of $\CPalign$ drops to polynomial time (denoted as $\P$) on this class. While we can generate the synchronous product of systems in polynomial time, we make use of the fact that this also holds for the synchronous product of reachability graphs (which adapts the idea of \autocite[Definition~8.6]{Winskel1984}).
	\begin{definition}[Synchronous Product of Reachability Graphs]\label{def:synchronous_reachability_graph}
		Let $R_1\coloneq(\M_1,M_{0,1},T_1,A_1,\ell_1)$ and $R_2\coloneq(\M_2,M_{0,2},T_2,A_2,\ell_2)$ be the reachability graphs of two systems with the labeling functions $\ell_1\colon T_1\to\Sigmatau$ and $\ell_2\colon T_2\to\Sigmatau$ over some alphabet $\Sigma$. Furthermore, let $\nomove\notin\Sigma,T_1,T_2$ be a distinguished symbol. Their synchronous product $R_1\otimes R_2\coloneq(\M,M_0,T,A,\ell)$ with the labeling function $\ell\colon T\to\Sigmatau$ is defined such that
		\begin{itemize}
			\item $\M\coloneq\M_1\times\M_2$ and $M_0\coloneq M_{0,1}+M_{0,2}$,
			\item $T\coloneq\Set{(t_1,t_2)\in T_1\times T_2\given\ell_1(t_1)=\ell_2(t_2)}\cup(T_1\times\Set\nomove)\cup(\Set\nomove\times T_2)$,
			\item $A\coloneq\Set{a_1\oplus a_2\in A_1\times A_2\given\ell_1(\pi_2(a_1))=\ell_2(\pi_2(a_2))}\cup(A_1\times \M_2^\circ)\cup(\M_1^\circ\times A_2)$\\
			where $a\oplus a'$ denotes the arc defined by
			\begin{equation*}
				(M,t,\bar{M})\oplus(M',t',\bar{M}')\coloneq(M+M',(t_1,t_2),\bar{M}+\bar{M}')
			\end{equation*}
			and $\M^\circ$ denotes the set of \emph{\enquote{no-move} arcs} defined by $\M^\circ\coloneqq\Set{(M,\nomove,M)\given M\in\M}$, and
			\item $(t_1,t_2)\mapsto\ell(t_1,t_2)\coloneqq\begin{cases}
				(\ell_1(t_1),\ell_2(t_2)) & t_1\in T_1 \wedge t_2\in T_2, \\
				(\ell_1(t_1),\nomove)     & t_2\notin T_2,                \\
				(\nomove,\ell_2(t_2))     & t_1\notin T_1.
			\end{cases}$
		\end{itemize}
	\end{definition}
	According to \zcref{def:reachability_graph,def:synchronous_product,def:synchronous_reachability_graph,cor:synchronous_reachability_set}, it does not matter if we compute the synchronous product or construct the reachability graph first.
	
	\begin{proposition}\label{prp:reachability_graph_product_equivalence}
		Given two systems $S_1$ and $S_2$, the synchronous product of their reachability graphs $\reachGraph{S_1}\otimes\reachGraph{S_2}$ is identical to the reachability graph of their synchronous product $\reachGraph{S_1\otimes S_2}$, i.e., $\reachGraph{S_1}\otimes\reachGraph{S_2}=\reachGraph{S_1\otimes S_2}$.
	\end{proposition}
	Therefore, we refer to $\reachGraph{S_1\otimes S_2}$ and $\reachGraph{S_1}\otimes\reachGraph{S_2}$, respectively, as the \emph{synchronous reachability graph} of the systems $S_1$ and $S_2$. Based on the construction of the synchronous reachability graph, we can now show the following result.
	\begin{theorem}
		\label{thm:align_ssystem_p}%
		On the class of safe and sound S-systems and on the class of live, safe S-systems, $\CPalign$ is in $\P$.
	\end{theorem}
	\begin{proof}
		Finding an optimal alignment between a trace $\sigma$ and an easy-sound system $S$ is equivalent to finding a shortest path on $\reachGraph{\mathcal T(\sigma)\otimes S}$ between the initial and the final state. According to \zcref{prp:reachability_graph_product_equivalence}, $\reachGraph{\mathcal T(\sigma)\otimes S}=\reachGraph{\mathcal{T}(\sigma)}\otimes\reachGraph{S}$ where the synchronous product and $\reachGraph{\mathcal{T}(\sigma)}$ are polynomial-time constructible (\zcref{def:synchronous_reachability_graph,cor:reachability_graph_p}). This, however, does not generally hold for $\reachGraph{S}$. But, if $S$ is a safe S-system that is either live or sound, its initial marking consists of exactly one token by \zcref{lem:safe_ssystem_single_token}. Due to \zcref{prop:reachability_graph_p}, its reachability graph can therefore be constructed in polynomial time. Finally, the shortest-path problem is known to be in~$\P$~\autocite[cf.][]{Ford1956}, and therefore $\CPalign$ is also in $\P$ on this class.
	\end{proof}
	
	The same argument applies whenever the number of tokens is bounded by a constant.
	\begin{proposition}
		\label{prop:align_k_tokens_p}%
		Let $m\in\N$ be a constant. On the class of systems in which every reachable marking consists of at most $m$ tokens, $\CPalign$ is in $\P$.
	\end{proposition}
	\begin{proof}
		Such a system $S=(N,M_\init,M_\final)$ with $N=(P,T,F,\ell)$ has at most $(\abs P+1)^m$ reachable markings, so $\reachGraph{S}$ can be constructed in polynomial time by a breadth-first search. The claim then follows as in the proof of \zcref{thm:align_ssystem_p}.
	\end{proof}
	In S-systems, the number of tokens is determined by the initial marking. For live S-systems, a bound $b$ even implies that there are at most $b$ tokens~\autocite[Theorem~3.5]{DeselE1995}, so \zcref{prop:align_k_tokens_p} shows that $\CPalign$ is in $\P$ on live, $b$-bounded S-systems for every constant $b$.
	
	We remark that for S-systems with a single token, which are basically standard directed graphs, the reachability problem is $\NL$-complete~\autocite{Jones1975}. More generally, $\CPreach$ is $\NL$-complete on the class of systems in which every reachable marking consists of at most $m$ tokens for a constant $m$, since such markings can be stored in logarithmic space. The same holds for safe S-systems: two tokens of a safe S-system cannot reach a common place (otherwise, both could be moved there), so reachability decomposes into reachability problems for single tokens. For safe S-systems that are live or sound, $\CPreach$ even becomes trivial: by \zcref{lem:safe_ssystem_single_token}, there is only a single token, and since there are no dead transitions, every marking that consists of exactly one token is reachable. Beyond these observations, we do not go into further details, since our foremost interest lies in classes which are at least as expressive as T-systems.
	
	\section{Alignments on Acyclic Petri Nets}
	\label{sec:acyclic:systems}
	Finally, let us turn our attention to the case of \emph{acyclic} systems $(N,M_0)$, i.e., systems where the underlying Petri net $N$ is acyclic.
	We have already seen in \zcref{cor:align_tsystem_nphard} that the alignment problem for acyclic systems is $\NP$-hard.
	What remains is to show that $\CPalign$ is also in $\NP$ on this class, even if we drop the safeness assumption.
	
	To this end, we recall the \emph{marking equation} of a Petri net system. We associate to a Petri net $N=(P,T,F,\ell)$ its \emph{incidence matrix} $\mat{N}\in\{-1,0,1\}^{P\times T}$.
	\begin{definition}[Incidence Matrix of a Petri Net]
		Let $N=(P,T,F,\ell)$ be a Petri net such that $P\neq\emptyset$ and $T\neq\emptyset$. Its \emph{incidence matrix} $\mat{N}\in\Set{-1,0,+1}^{P\times T}$ is defined by
		\begin{equation*}
			\mat{N}(p,t)\coloneq\begin{cases}
				-1 & \text{if } (p,t)\in F\wedge (t,p)\notin F, \\
				+1 & \text{if } (p,t)\notin F\wedge (t,p)\in F, \\
				\phantom+0  & \text{otherwise}.
			\end{cases}
		\end{equation*}
	\end{definition}
	Moreover, we associate with every firing sequence $\sigma\in T^*$ its so-called \emph{Parikh vector} $\vec\sigma\in\N^T$, i.e., the vector which abstracts from the order of transitions in $\sigma$ and only counts how many times each transition occurs.
	Then, if $\sigma$ is a firing sequence of the system $(N,M_0)$ which leads to a marking $M$, then the following \emph{marking equation} holds:
	\begin{equation*}
		M=M_0+\mat{N}\vec\sigma.
	\end{equation*}
	In this way, we can use the marking equation as a necessary condition for the reachability of a marking $M$ from~$M_0$: if there is no solution $x\in\N^T$ of the (linear) equation $M=M_0+\mat{N}x$, then $M$ is not reachable from $M_0$.
	Unfortunately, for general Petri nets, the converse is not true: even if we find a solution $x\in\N^T$ to the marking equation, then this does not mean that this solution is also \emph{realizable}, i.e., that there is an actual firing sequence $\sigma$ enabled at $M_0$ with $\vec\sigma=x$.
	
	Luckily, for acyclic systems, the situation changes. Here, the marking equation is not only a necessary but also a sufficient condition for reachability. In particular, a solution $x\in\N^T$ to the marking equation is always realizable~\autocite{HiraishiI1988,Murata1989}.
	This means that in acyclic systems, solutions to the marking equation and enabled firing sequences are in one-to-one correspondence.
	
	With this result in mind, it is easy to show that $\CPalign$ is in $\NP$ on acyclic systems. Given a trace $\sigma$, we first construct the synchronous product of the system $(N,M_0)$ and the trace system $\mathcal T(\sigma)$, which is again an acyclic system.
	Then, in order to determine whether the alignment costs between the system $(N,M_0)$ and $\sigma$ do not exceed a threshold $k\in\costRange$ with respect to some cost function $c$, we simply check if the final marking $M_\final'$ of the synchronous product net $N'$ is reachable from the initial marking $M_\init'$ via a firing sequence $\sigma'$ of cost at most $k$.
	Since the system is acyclic, this reduces to guessing a vector $x\in\N^T$ of the marking equation $M_\final'=M_\init'+\mat{N'}x$ such that $\sum_{t\in T} c(t)\cdot x(t)\leq k$. This immediately yields an integer linear program which can be solved in $\NP$ and thus shows that $\CPalign$ is in $\NP$ on acyclic systems. Although the given acyclic system might not be easy-sound, this is easily accomplished using a similar construction as already shown for \zcref{lem:mincostreach_to_align}.
	
	With the $\NP$-hardness result in \zcref{cor:align_tsystem_nphard}, we can conclude that $\CPalign$ is $\NP$-complete on acyclic systems, even when they are safe.
	\begin{theorem}
		\label{thm:align_acyclic_np}%
		On the class of acyclic systems, $\CPalign$ is $\NP$-complete. $\NP$-hardness holds even for safe and sound acyclic systems.
	\end{theorem}
	
	\section{Conclusion}
	\label{sec:conclusion}
	
	While alignments play an important role in almost all process mining applications, the algorithmic costs for aligning a model and a trace are extremely high.
	As a matter of fact, alignments remain an algorithmic challenge even for state-of-the-art process mining tools.
	In this article, we proved that the high computational costs for computing alignments are unavoidable: an efficient algorithm for alignments on sound workflow nets does not exist (assuming $\P\neq\PSPACE$).
	However, instead of giving up, we further showed that it is worth taking a closer look at the structure of alignments.
	Specifically, we derived new algorithmic guarantees for important model classes, such as the class of sound, free-choice workflow nets which, in turn, includes all process trees, a class highly relevant in practice.
	
	\begin{table}
		\centering
		\caption{Overview of our results for the complexity of the alignment problem ($\CPalign$) including a comparison with the complexity of the reachability problem ($\CPreach$).}
		\label{tab:overview:results}
		\zcsetup{abbrev=true}
		\begin{booktabs}{colspec={lccl},colsep=5pt,cell{1}{1}={r=2}{l},cell{1}{2}={c=3}{c},cell{2}{3}={c=2}{c}}
			\toprule
			Model Class                                        &     Complexity of     &                    &                                                           \\ \cmidrule[lr]{2-4}
			                                                   &      $\CPreach$       &     $\CPalign$     &                                                           \\ \midrule
			Petri Net Systems                                  & $\Ackermann$-complete & $\Ackermann$-hard  & \zcref{cor:align_general_ackermann_hard}                  \\
			\qquad Bounded                                     &  $\PSPACE$-complete   & $\PSPACE$-complete & \zcref{cor:align_bounded_pspace_complete}                 \\
			\qquad Safe                                        &  $\PSPACE$-complete   & $\PSPACE$-complete & \zcref{thm:align_safe_pspace_complete}                    \\
			\qquad Cyclic, Live, Safe                          &  $\PSPACE$-complete   & $\PSPACE$-complete & \zcref{cor:align_cyclic_pspace_complete}                  \\
			\qquad Live, Bounded, Free-Choice                  &    $\NP$-complete     &   $\NP$-complete   & \zcref{cor:align_lfbc_npcomplete}                         \\
			\qquad Acyclic                                     &    $\NP$-complete     &   $\NP$-complete   & \zcref{thm:align_acyclic_np}                              \\
			\qquad Safe, Acyclic                               &    $\NP$-complete     &   $\NP$-complete   & \zcref{thm:align_acyclic_np}                              \\
			\qquad Bounded, Conflict-Free                      &        $\in\P$        &     $\NP$-hard     & \zcref{cor:align_safe_conflictfree_nphard}                \\
			\qquad Safe, Conflict-Free                         &        $\in\P$        &     $\NP$-hard     & \zcref{cor:align_safe_conflictfree_nphard}                \\
			\qquad Cyclic, Live, Bounded, Free-Choice          &        $\in\P$        &   $\NP$-complete   & \zcref{cor:align_lfbc_npcomplete}                         \\
			\qquad Cyclic, Live, Safe, Free-Choice             &        $\in\P$        &   $\NP$-complete   & \zcref{cor:align_lfbc_npcomplete}                         \\
			\qquad Constant Number of Tokens                   &    $\NL$-complete     &      $\in\P$       & \zcref{prop:align_k_tokens_p}                             \\ \addlinespace
			Sound Workflow Nets                                &    $\PSPACE$-hard     &   $\PSPACE$-hard   & \zcref{lem:wfnets:pspace}                                 \\
			\qquad Safe                                        &  $\PSPACE$-complete   & $\PSPACE$-complete & \zcref{thm:align_wfnets_pspace_complete}                  \\
			\qquad Free-Choice                                 &        $\in\P$        &   $\NP$-complete   & \zcref{cor:align_sound_fc_wfnets_npcomplete}              \\ \addlinespace
			POWL Models                                        &        $\in\P$        &   $\NP$-complete   & \zcref{cor:align_powl_npcomplete}                         \\ \addlinespace
			Process Trees                                      &        $\in\P$        &   $\NP$-complete   & \zcref{thm:pt_npcomplete}                                 \\
			\qquad Restricted to $\Set{\ptseq,\ptxor,\ptloop}$ &        $\in\P$        &      $\in\P$       & \autocite{SchwanenPA2025,SchwanenPA2025a,SchwanenPA2025c} \\
			\qquad with Unique Labels                          &        $\in\P$        &      $\in\P$       & \zcref{thm:pt:unique}                                     \\ \addlinespace
			T-Systems                                          &        $\in\P$        &     $\NP$-hard     & \zcref{cor:align_tsystem_nphard}                          \\
			\qquad Bounded                                     &        $\in\P$        &   $\NP$-complete   & \zcref{thm:align_bounded_tsystem_np}                      \\
			\qquad Safe, Live/Sound                            &        $\in\P$        &   $\NP$-complete   & \zcref{thm:align_bounded_tsystem_np}                      \\ \addlinespace
			S-Systems                                          &        $\in\P$        &     $\NP$-hard     & \zcref{thm:align_bounded_ssystem_np}                      \\
			\qquad Bounded                                     &        $\in\P$        &   $\NP$-complete   & \zcref{thm:align_bounded_ssystem_np}                      \\
			\qquad Bounded, Live/Sound                         &        $\in\P$        &   $\NP$-complete   & \zcref{thm:align_bounded_ssystem_np}                      \\
			\qquad Safe                                        &    $\NL$-complete     &   $\NP$-complete   & \zcref{thm:align_safe_ssystem_np}                         \\
			\qquad Safe, Live/Sound                            &     $\in\Lspace$      &      $\in\P$       & \zcref{thm:align_ssystem_p}                               \\ \bottomrule
		\end{booktabs}
		\medskip
		
		For bounded Petri nets, we assume the bound $b$ to be part of the input (given in unary notation).
	\end{table}
	When we started this project, our intuition was that the alignment problem and the reachability problem are of the same complexity level. While we confirmed this for safe Petri nets, we found a much more interesting picture on tamer model classes. Most remarkably, on live, bounded, free-choice systems, reachability~\autocite{Esparza1998a} and alignments are both $\NP$-complete. However, if we add cyclicity (i.e., the assumption that the initial marking is reachable from any other marking), reachability drops to $\P$ while the alignment problem remains $\NP$-complete. In particular, this complexity gap also holds for sound free-choice workflow nets as well as simpler model classes like process trees (\zcref{thm:pt_npcomplete}) or Partially Ordered Workflow Language (POWL) models (\zcref{cor:align_powl_npcomplete}), a generalization of process trees---both being strict subclasses of sound free-choice workflow nets. Similarly, reachability is in $\P$ on T-systems and on safe conflict-free systems, whereas the alignment problem is $\NP$-hard on both classes and even $\NP$-complete on bounded T-systems.
	Interestingly, for S-systems the complexity of alignments even depends on the number of initial tokens: if it is bounded by a constant, in particular for a single token, alignments can be computed in polynomial time, whereas the alignment problem becomes $\NP$-complete if the number of tokens is part of the input.
	The intuitive reason is, as our proofs show, that multiple tokens bring concurrency back into the picture of the alignment problem. For live or sound S-systems, a single initial token is equivalent to safeness (\zcref{lem:safe_ssystem_single_token}).
	For process trees, we were able to show that alignments are in $\P$ when excluding the parallel operator~\autocite{SchwanenPA2025} or if the subtree rooted at a parallel operator is at least restricted to only have unique labels~\autocite{SchwanenPA2025a,SchwanenPA2025c}.
	\zcref[S]{tab:overview:results} summarizes our results for all model classes considered in this article and additionally compares the complexity of the alignment problem with that of the corresponding reachability problem.
	
	Our research also led to new contributions in Petri net theory. For example, we were able to generalize the well-known \emph{Shortest Sequence Theorem}~\autocite{DeselE1995a,DeselE1995} for \lbfc-systems from an existence statement of a shortest sequence to a universal statement, namely that every sequence can be shortened accordingly~(\zcref{thm:shortestst:tsystems:rev}). Since shortening a firing sequence in this way never increases its costs under non-negative cost functions, the result also applies to other optimization problems over firing sequences; for instance, it directly shows that $\CPcostReach$ is in $\NP$ on \lbfc-systems.
	
	There are many opportunities for further studies.
	Obviously, we would like to see if our theoretical results can be empirically verified. Also, based on our findings, we would like to assemble a set of \emph{hard} instances to analyze different algorithmic approaches in practice.
	For future research, we would also like to investigate to what extent incorporating additional information like event attributes \autocite{MannhardtLRA2016} or considering enabled activities given in translucent event logs \autocite{vanderAalst2019,BeyelSA2025} affect the complexity of alignments.
	On the more theoretical level, we want to dive deeper into the complexity structure of alignments, in particular, regarding the influence of different parameters of the models, e.g., we want to investigate the influence of \enquote{concurrency}, \enquote{non-observable behavior} and \enquote{decomposability} on the complexity of alignments using the machinery of parameterized complexity theory.
	Questions also include if restricting the labeling function has a positive effect for other model classes as well or if the bounds of the \emph{Shortest Sequence Theorem} are tight for sound free-choice workflow nets or can even be further improved.
	
	Another interesting direction is to investigate a generalized form of alignments, in which we also allow other distance metrics. So far, we have only considered a simple type of edit distance which is limited to insertions and deletions only. Although it is fairly easy to see that including substitutions would not affect the results presented in this article, the effect of including transpositions is an open research problem.
	
	\printbibliography
\end{document}